\documentclass[runningheads]{llncs}

\newif\iffullversion
\fullversiontrue
\usepackage[T1]{fontenc}
\usepackage[utf8]{inputenc}
\usepackage{microtype}

\usepackage{xspace}
\usepackage{amssymb}
\usepackage{mathtools}

\usepackage{cite}
\usepackage{graphicx}
\usepackage{proof}
\usepackage{polytable}
\usepackage{alltt}
\usepackage{calc}

\usepackage{wrapfig}
\usepackage{dcolumn}
\newcolumntype{d}[1]{D{.}{.}{#1}}

\spnewtheorem{simplification}{Simplification}{\bfseries}{\rmfamily}
\spnewtheorem{fact}{Fact}{\bfseries}{\rmfamily}
\spnewtheorem*{fact*}{Fact}{\bfseries}{\rmfamily}
\spnewtheorem*{theorem*}{Theorem}{\bfseries}{\rmfamily}

\makeatletter
\RequirePackage[bookmarks,unicode,colorlinks=true]{hyperref}\def\@citecolor{blue}\def\@urlcolor{blue}\def\@linkcolor{blue}
\iffullversion\else
\def\orcidID#1{\smash{\href{http://orcid.org/#1}{\protect\raisebox{-1.25pt}{\protect\includegraphics{orcid_color.eps}}}}}
\fi
\makeatother

\usepackage{cleveref}
\crefname{appendix}{Appendix}{Appendix}
\Crefname{appendix}{Appendix}{Appendix}
\crefname{section}{Sect.\@}{Sect.\@}
\Crefname{section}{Section}{Sections}
\crefname{figure}{Fig.\@}{Fig.\@}
\Crefname{figure}{Figure}{Figures}
\crefname{example}{Example}{Examples}
\Crefname{example}{Example}{Examples}
\crefname{lemma}{Lemma}{Lemmas}
\crefname{corollary}{Corollary}{Corollaries}
\crefname{theorem}{Theorem}{Theorems}
\crefname{fact}{Fact}{Facts} 

\makeatletter
\@ifundefined{mathindent}{\newdimen\mathindent\mathindent\leftmargini}{}\makeatother
\newlength{\blanklineskip}
\allowdisplaybreaks

\makeatletter
\usepackage{latexsym}
\usepackage{xparse}
\usepackage{xifthen}

\def\qedmath{\relax\ifmmode
      \@badmath
   \else
      \begin{trivlist}\@beginparpenalty\predisplaypenalty
      \@endparpenalty\postdisplaypenalty
      \item[]\leavevmode
      \hbox to\linewidth\bgroup\relax\hfil$\displaystyle
      \bgroup
   \fi
}
\def\endqedmath{\relax\ifmmode
      \egroup $\hfil \egroup\\[-\baselineskip]
      \hbox{}\hfill\llap{\phantom{0}\qed}
      \end{trivlist}
   \else
      \@badmath
   \fi
}
\def\qedflmath{\relax \ifmmode \@badmath
  \else\begin{trivlist}\@beginparpenalty\predisplaypenalty
     \@endparpenalty\postdisplaypenalty
    \item[]\leavevmode \hb@xt@\linewidth\bgroup $\m@th\displaystyle \hskip\mathindent\bgroup
  \fi}
\def\endqedflmath{\relax\ifmmode
    \egroup $\egroup \\[-\baselineskip]\hskip\linewidth\llap{\qed}
    \end{trivlist}\else \@badmath
  \fi}
\makeatother

\usepackage{color}

\newcommand{\todo}[1]{\textcolor{red}{\textbf{[TODO: #1]}}}

\newcommand{\mi}[1]{\ensuremath{\mathit{#1}}}
\newcommand{\mr}[1]{\ensuremath{\mathrm{#1}}}
\newcommand{\ms}[1]{\ensuremath{\mathsf{#1}}}

\newcommand{\mc}[1]{\ensuremath{\mathcal{#1}}}

\newcommand{\foreignWord}[1]{#1}
\newcommand{\eg}{\foreignWord{e.g.}\xspace}
\newcommand{\ie}{\foreignWord{i.e.}\xspace}

\newcommand{\etal}{\foreignWord{et al.}\xspace}

\newcommand{\Case}[1]{\textbf{Case}: #1.}
\newcommand{\CaseP}[1]{\paragraph{\normalfont\textbf{Case}: #1.}}

\newcommand{\Gl}{\lambda}
\newcommand{\Ga}{\alpha}
\newcommand{\Gb}{\beta}

\newcommand{\Gs}{\sigma}

\newcommand{\Gt}{\tau}
\newcommand{\Gth}{\theta}

\newcommand{\Gg}{\gamma}
\newcommand{\GG}{\Gamma}

\newcommand{\GD}{\Delta}
\newcommand{\Gp}{\pi}

\newcommand{\Gz}{\zeta}
\newcommand{\Gu}{\upsilon}

\newcommand{\ox}{\otimes}

\newenvironment{tarray}[2][c]{\settowidth{\dimen1}{${} = {}$}\setlength{\arraycolsep}{0.125\dimen1}
\hspace{-0.125\dimen1}\array[#1]{#2}\relax
}
{\endarray\hspace{-0.125\dimen1}
}
\newcommand{\bb}{\tarray{lllll}}
\newcommand{\bbc}{\tarray{c}}
\newcommand{\bbt}{\tarray[t]{lllll}}

\newcommand{\ee}{\endtarray}

\newcommand{\pto}[1]{\to_{#1}}

\newcommand{\To}[0]{\Rightarrow}
\newcommand{\From}{\Leftarrow}

\newcommand{\DOLLAR}{\ensuremath{\mathbin{\texttt{\$}}}}

\newcommand{\BIND}{\ensuremath{\mathbin{\texttt{>\!\!>\!=}}}}

\newcommand{\key}[1]{\ensuremath{\mathbf{#1}}}

\newcommand{\CASE}{\key{case}}

\newcommand{\A}{\;}

\newcommand{\OF}{\key{of}}

\newcommand{\LET}{\key{let}}
\newcommand{\IN}{\key{in}}

\newcommand{\shortequals}{{\rlap{\rule[0.14em]{.25em}{0.5pt}}{\rule[0.34em]{.25em}{0.5pt}}}
}

\newcommand{\Lub}{\sqcup}

\newcommand{\doubleequals}{\ensuremath{\mathrel{\shortequals\,\shortequals}}}
\newcommand{\ndoubleequals}{\ensuremath{\not\expandafter\doubleequals}}

\usepackage{stmaryrd}

\newcommand{\Empty}{{\varepsilon}}

\makeatletter
\newcounter{pe@nocounter}
\newcommand{\peano}[1]{\setcounter{pe@nocounter}{#1}
\pe@no}
\newcommand{\pe@no}{\ifnum\value{pe@nocounter}>1
\addtocounter{pe@nocounter}{-1}\con{S} \A (\expandafter\pe@no)\else \ifnum\value{pe@nocounter}>0\con{S} \A \con{Z}\else \con{Z}\fi \fi }

\newcommand{\set}{\@ifstar\mysetStar\mysetNoStar}
\newcommand{\mysetStar}[1]{\{#1\}}
\newcommand{\mysetNoStar}[1]{\left\{#1\right\}}
\newcommand{\sem}{\@ifstar\mySemStar\mySemNoStar}
\newcommand{\mySemStar}[1]{{\llbracket #1 \rrbracket}}
\newcommand{\mySemNoStar}[1]{{\left\llbracket #1 \right\rrbracket}}
\makeatother

\newcommand{\dom}{\mathsf{dom}}

\newcommand{\fv}{\mathsf{fv}}

\newcommand{\rname}[1]{\text{\textsc{#1}}}
\newcommand{\ninfer}[3]{\infer[\!\!\rname{#1}]{#2}{#3}}

\newcommand{\con}[1]{\ms{#1}}
\newcommand{\var}[1]{\mi{#1}}

\newcommand{\DEq}{\vcentcolon\vcentcolon=}

\newcommand{\V}[1]{\overline{#1}}

\newcommand{\PROMPT}{\texttt{>}}

\newcommand{\mynewpar}[1]{{\parskip=0pt\parindent=0pt\par\vskip #1\noindent}}

\newcommand{\mycodestart}{\(}
\newcommand{\mycodeend}{\)}
\newcommand{\mycolumn}[1]{\column{#1}{@{}>{\mycodestart}l<{\mycodeend}@{}}}

\newcommand{\myprompt}[1]{\column{#1}{@{\PROMPT~{}}>{\mycodestart~}l<{\mycodeend}@{}}}

\newcommand{\declarecolumns}{\mycolumn{a}\mycolumn{b}\mycolumn{c}\mycolumn{d}\mycolumn{e}\mycolumn{f}\mycolumn{g}\mycolumn{h}\mycolumn{i}\mycolumn{j}\mycolumn{k}\mycolumn{l}\mycolumn{m}\mycolumn{n}\mycolumn{o}\mycolumn{p}\mycolumn{q}\mycolumn{r}\mycolumn{s}\mycolumn{t}\mycolumn{u}\mycolumn{v}\mycolumn{w}\mycolumn{x}\mycolumn{y}\mycolumn{z}\mycolumn{A}\mycolumn{B}\mycolumn{C}\mycolumn{D}\mycolumn{E}\mycolumn{F}\mycolumn{G}\mycolumn{H}\mycolumn{I}\mycolumn{J}\mycolumn{K}\mycolumn{L}\mycolumn{M}\mycolumn{N}\mycolumn{O}\mycolumn{P}\mycolumn{Q}\mycolumn{R}\mycolumn{S}\mycolumn{T}\mycolumn{U}\mycolumn{V}\mycolumn{W}\mycolumn{X}\mycolumn{Y}\mycolumn{Z}\mycolumn{L1}\mycolumn{L2}\mycolumn{L3}\mycolumn{L4}\mycolumn{L5}\mycolumn{L6}\mycolumn{W1}\mycolumn{W2}\mycolumn{W3}\mycolumn{W4}\mycolumn{W5}\mycolumn{R1}\mycolumn{R2}\mycolumn{R3}\mycolumn{R4}\mycolumn{R5}\myprompt{CL}}

\newenvironment{code}{\mynewpar\abovedisplayskip \advance\leftskip\mathindent \pboxed \declarecolumns \mycolumn{definetab}}{\endpboxed \mynewpar\belowdisplayskip \ignorespacesafterend }

\makeatletter
\NewDocumentEnvironment{minicode}{O{t}m}{\minipage[#1]{#2}\code}{\endcode\endminipage}

\makeatother

\definecolor{col-rev}{rgb}{0.55,0.05,0.08}
\NewDocumentCommand\revdecorate{m}{{\color{col-rev}#1}}
\NewDocumentCommand\revmark{m}{{#1}^{\bullet}}
\NewDocumentCommand\REV{m}{\ifthenelse{\equal{#1}{}}{\revmark{\revdecorate{(-)}}}{\revmark{\revdecorate{#1}}}}

\newcommand{\Subst}[2]{#1 \mapsto #2}
\newcommand{\VSubst}[2]{\V{\Subst{#1}{#2}}}

\newcommand{\eqty}{\sim}
\newcommand{\Disjoint}[2]{\ensuremath{#1 \cap #2 = \emptyset}}

\definecolor{myyellow}{rgb}{1,1,0.3}
\definecolor{mygray}{gray}{0.88}
\newcommand{\mathhighlight}[1]{\colorbox{mygray}{$\displaystyle#1$}}

\newcommand{\ftv}{\mathsf{ftv}}
\newcommand{\fuv}{\mathsf{fuv}}

\newlength\vdashHeight
\newlength\headHeight
\newcommand{\vdashI}{\settoheight{\vdashHeight}{\(\vdash\)}\settoheight{\headHeight}{\scalebox{0.7}{\(\blacktriangleright\)}}\advance \vdashHeight by -\headHeight\relax \divide  \vdashHeight by 2\relax \mathrel{{\vdash}\mkern-5mu\raisebox{\vdashHeight}{\scalebox{0.7}{\(\blacktriangleright\)}}}}

\newcommand{\OTy}[4]{ #1 ; #2 \vdash_\mathrm{o} #3 : #4}
\newcommand{\NTy}[5]{#5 ; #1 ; #2 \vdash #3 : #4}
\newcommand{\WTy}[5]{#5 ; #1 ; #2 \vdash_\mr{wp} #3 : #4}
\newcommand{\NTyP}[2]{ #1 \vdash #2 }
\newcommand{\ITy}[5]{#1 \vdashI #3 : #4 \leadsto #2 ; #5 }
\newcommand{\ITyP}[2]{#1 \vdashI #2 }
\newcommand{\ISimpl}[4]{ #1 \vdashI_\mathrm{simp} #2 \leadsto #3; #4 } \newcommand{\Simpl}[4]{#1 \vdashI_\mathrm{simp} #2 \leadsto #3; #4 }
\newcommand{\SimplE}[5]{#1 \vdashI_\mathrm{simp}^{#5} #2 \leadsto #3; #4 }

\newcommand{\ofUse}{\mapsto}
\newcommand{\Used}[2]{{#1}^{#2}}

\newcommand{\rVar}{\rname{Var}}
\newcommand{\rAbs}{\rname{Abs}}
\newcommand{\rApp}{\rname{App}}

\newcommand{\rCon}{\rname{Con}}

\newcommand{\rCase}{\rname{Case}}

\newcommand{\rEmpty}{\rname{Empty}}
\newcommand{\rBind}{\rname{Bind}}
\newcommand{\rBindA}{\rname{BindA}}

\newcommand{\rUnify}{\rname{S-Uni}\xspace}

\newcommand{\rEntailEq}{\rname{S-Eq}\xspace}
\newcommand{\rEntail}{\rname{S-Entail}\xspace}
\newcommand{\rRem}{\rname{S-Rem}\xspace}

\newcommand{\Elim}[2]{\ms{elim}(\exists #1.#2)}

\begin{document}
\title{Modular Inference of Linear Types for Multiplicity-Annotated Arrows}
\author{Kazutaka Matsuda\inst{1}
}

\institute{Tohoku University, Sendai 980-8579, Japan\\
\email{kztk@ecei.tohoku.ac.jp}
}

\authorrunning{K. Matsuda}

\maketitle              \begin{abstract}
Bernardy \etal [2018] proposed a linear type system $\lambda^q_\to$ as a core type system of Linear Haskell. 
In the system, linearity is represented by annotated arrow types $A \to_m B$, where $m$ denotes 
the multiplicity of the argument. Thanks to this representation, existing non-linear code typechecks as it is,
and newly written linear code can be used with existing non-linear code in many cases. 
However, little is known about the type inference of $\lambda^q_\to$. Although the Linear Haskell implementation is equipped with type inference, its algorithm has not been formalized,
and the implementation often fails to infer principal types, especially for higher-order functions. 
In this paper, based on \textsc{OutsideIn(X)} [Vytiniotis \etal, 2011], we propose an inference system for 
a rank 1 qualified-typed variant of $\lambda^q_\to$, which infers principal types. 
A technical challenge in this new setting is to deal with ambiguous types inferred by naive qualified typing. 
We address this ambiguity issue through quantifier elimination and demonstrate the effectiveness of the approach with examples. 
\end{abstract}

\section{Introduction}
\label{sec:intro}
\label{sec:introduction}

Linearity is a fundamental concept in computation and has many applications. 
For example, if a variable is known to be used only once, it can be freely inlined 
 without any performance regression~\cite{TurnerWM95}. 
In a similar manner, destructive updates are safe for such values without the risk of breaking
referential transparency~\cite{Wadler90LT}.
Moreover, linearity is useful for writing transformation on data that cannot be copied or discarded for various reasons, 
including reversible computation~\cite{Janus,YoAG11} and quantum computation~\cite{AltenkirchG05,SelingerV06}.
Another interesting application of linearity is that it helps to bound the complexity of programs~\cite{DBLP:journals/tcs/AehligBHS04,BaillotT09,DBLP:journals/tcs/GirardSS92}

Linear type systems use types to enforce linearity. One way to design a linear type system is based on Curry-Howard isomorphism to linear logic.
For example, in Wadler~\cite{Wadler93}'s type system, 
functions are linear in the sense that their arguments are used exactly once, and 
any exception to this must be marked by the type operator ($!$). 
Such an approach is theoretically elegant but cumbersome in programming; a program usually contains both linear and unrestricted code, 
and many manipulations concerning ($!$) are required in the latter and around the interface between the two.
Thus, there have been several proposed approaches for more practical linear type systems~\cite{MaZZ10,Morris16,TovP11,BeBNJS18}. 

Among these approaches, a system called $\Gl^q_\to$, the core type system of Linear Haskell, 
stands out for its ability to have linear code in large unrestricted code bases~\cite{BeBNJS18}. 
With it, existing unrestricted code in Haskell typechecks in Linear Haskell without modification, and if one desires, 
some of the unrestricted code can be replaced with linear code, again without any special programming effort. 
For example, one can use the function $\var{append}$ in an unrestricted context as 
$\Gl x. \var{tail} \A (\var{append} \A x \A x)$, regardless of 
whether $\var{append}$ is a linear or unrestricted function. 
This is made possible by their representation of linearity. Specifically, they annotate function type
with its argument's multiplicity (``linearity via arrows''~\cite{BeBNJS18}) as $A \pto{m} B$, where $m = 1$ means that the function of the type uses its argument linearly, and 
$m = \omega$ means that there is no restriction in the use of the argument, which includes all non-linear standard Haskell code. 
In this system, linear functions can be used in an unrestricted context if 
their arguments are unrestricted. 
Thus, there is no problem in using 
$\var{append} : \con{List} \A A \pto{1} \con{List} \A A \pto{1} \con{List} \A A$ as above, provided that $x$ is unrestricted.
This {promotion} of linear expressions to unrestricted ones is difficult 
in other approaches~\cite{MaZZ10,Morris16,TovP11} (at least in the absence of bounded kind-polymorphism), where linearity is a property of a type (called ``linearity via kinds'' in \cite{BeBNJS18}).

However, as far as we are aware, little is known about \emph{type inference} for $\Gl^q_\to$. 
It is true that Linear Haskell is implemented as a fork\footnote{\url{https://github.com/tweag/ghc/tree/linear-types}} of the Glasgow Haskell Compiler (GHC), which of course comes with type inference.
However, the algorithm has not been formalized and has limitations due to a lack of proper handling of multiplicity constraints. 
Indeed, Linear Haskell gives up handling complex constraints on multiplicities such as those with multiplications $p \cdot q$; 
as a result, Linear Haskell sometimes fails to infer principal types, especially for higher-order functions.\footnote{Confirmed for commit \texttt{1c80dcb424e1401f32bf7436290dd698c739d906} at May 14, 2019.
\iffullversion See \cref{sec:limitation-LinearHaskell} for more details. \fi}
This limits the reusability of code. For example, 
Linear Haskell cannot infer an appropriate type for function composition to allow it to compose both linear and unrestricted functions.

A classical approach to have both separated constraint solving that works well with the usual unification-based typing 
and principal typing (for a rank 1 fragment) is qualified typing~\cite{QualifiedTypes}.
In qualified typing, constraints on multiplicities are collected, and then a type is qualified with it to obtain a principal type. 
Complex multiplicities are not a problem in unification as they are handled by a constraint solver.
For example, consider $\var{app} = \Gl f. \Gl x. f \A x$. 
Suppose that $f$ has type $a \pto{p} b$, and $x$ has type $a$ (here we focus only on multiplicities). 
Let us write the multiplicities of $f$ and $x$ as $p_f$ and $p_x$, respectively. 
Since $x$ is passed to $f$, there is a constraint that the multiplicity $p_x$ of $x$ must be $\omega$ if the multiplicity $p$ of the $f$'s argument also is.
In other words, $p_x$ must be no less than $p$, which is represented by inequality $p \le p_x$ under the ordering $1 \le \omega$. 
(We could represent the constraint as an equality $p_x = p \cdot p_x$, but using inequality is simpler here.)
For the multiplicity $p_f$ of $f$, there is no restriction because $f$ is used exactly once; linear use is always legitimate 
even when $p_f = \omega$.
As a result, we obtain the inferred type $\forall p\,p_f\,p_x\,a\,b.\: p \le p_x \To (a \pto{p} b) \pto{p_f} a \pto{p_x} b$ for $\var{app}$. 
This type is a principal one; it is intuitively because only the constraints that are needed for typing $\Gl f. \Gl x. f \A x$ are gathered.
Having separate constraint solving phases itself is rather common in the context of linear typing~\cite{TurnerWM95,WansbroughJ99,Mogensen97,IgarashiK00,GhicaS14,BaillotH10,BaillotT05,Morris16,GanTM15}.
Qualified typing makes the constraint solving phase local and gives the principal typing property that makes typing modular. 
In particular, in the context of linearity via kinds, 
qualified typing is proven to be effective~\cite{Morris16,GanTM15}.

As qualified typing is useful in the context of linearity via kinds, one may expect that it also works well 
for linearity via arrows such as $\Gl^q_\to$.
However, naive qualified typing turns out to be impractical for $\Gl^q_\to$ because it tends to infer ambiguous types~\cite{QualifiedTypes,StuckeyS05}.
As a demonstration, consider a slightly different version of $\var{app}$
defined as $\var{app}' = \Gl f. \Gl x. \var{app} \A f \A x$. Standard qualified typing~\cite{QualifiedTypes,VytiniotisJSS11} infers the type 
\[
\forall q\,q_f\,q_x\,p_f\,p_x\,a\,b.\: (q \le q_x \wedge q_f \le p_f \wedge q_x \le p_x) \To (a \pto{q} b) \pto{p_f} a \pto{p_x} b
\]
by the following steps:
\begin{itemize}
 \item The polymorphic type of $\var{app}$ is instantiated to $(a \pto{q} b) \pto{q_f} a \pto{q_x} b$ and 
       yields a constraint $q \le q_x$ (again we focus only on multiplicity constraints).
 \item Since $f$ is used as the first argument of $\var{app}$, $f$ must have type $a \pto{q} b$. 
   Also, since the multiplicity of $\var{app}$'s first argument is $q_f$, there is a restriction on the multiplicity of $f$, say $p_f$, 
   that $q_f \le p_f$. 
 \item Similarly, since $x$ is used as the second argument of $\var{app}$, $x$ must have type $a$,
   and there is a constraint on the multiplicity of $x$, say $p_x$, that $q_x \le p_x$. 
\end{itemize}
This inference is unsatisfactory, as 
the inferred type leaks internal details and is ambiguous~\cite{QualifiedTypes,StuckeyS05}
in the sense that one cannot determine $q_f$ and $q_x$ from an instantiation of $(a \pto{q} b) \pto{p_f} a \pto{p_x} b$.
Due to this ambiguity, the types of $\var{app}$ and $\var{app}$' are not judged as equivalent; 
in fact, the standard qualified typing algorithms~\cite{VytiniotisJSS11,QualifiedTypes} reject
$\var{app}' : \forall p\,p_f\,p_x\,a\,b.\: p \le p_x \To (a \pto{p} b) \pto{p_f} a \pto{p_x} b$. 
We conjecture that the issue of inferring ambiguous types is intrinsic to linearity via arrows 
because of the separation of multiplicities and types, unlike the case of linearity via kinds, where multiplicities are always associated with types. 
Simple solutions such as rejecting ambiguous types are not desirable as this case appears very often. Defaulting ambiguous variables (such as $q_f$ and $q_x$) to $1$ or $\omega$ is not a solution either because it loses principality in general.

In this paper, we propose a type inference method for a rank 1 qualified-typed variant of $\Gl^q_\to$, in which the ambiguity issue is 
addressed without compromising principality. Our type inference system is built on top of 
\textsc{OutsideIn(X)}~\cite{VytiniotisJSS11}, an inference system for qualified types used in GHC, which can handle 
local assumptions to support $\LET$, existential types, and GADTs. 
An advantage of using \textsc{OutsideIn(X)} is that it is parameterized over theory \textsc{X} of constraints.
Thus, applying it to linear typing boils down to choosing an appropriate \textsc{X}. 
We choose \textsc{X} carefully so that the representation of constraints is closed under quantifier elimination, which is the key to addressing the ambiguity issue.
Specifically, in this paper: 
\begin{itemize} 
 \item We present a qualified typing variant of a rank-1 fragment of $\Gl^q_\to$ without local definitions, 
   in which manipulation of multiplicities is separated from the standard unification-based typing (\cref{sec:lang}).
\item 
 We give an inference method for the system based on gathering constraints and solving them afterward~(\cref{sec:inference}).
 This step is mostly standard, except that we solve multiplicity constraints in time polynomial in their sizes.
\item We address the ambiguity issue by quantifier elimination 
       under the assumption that multiplicities do not affect runtime behavior
       (\cref{sec:defaulting}).
\item We extend our technique to local assumptions (\cref{sec:extensions}), 
   which enables $\LET$ and GADTs, by showing that the disambiguation in \cref{sec:defaulting} is compatible with \textsc{OutsideIn(X)}. 
\item We report experimental results using our proof-of-concept implementation (\cref{sec:evaluation}). 
The experiments show that the system can infer unambiguous principal types for selected functions from Haskell's \texttt{Prelude}, 
   and performs well with acceptable overhead.
\end{itemize}
Finally, we discuss related work~(\cref{sec:related}) and then conclude the paper (\cref{sec:conclusion}).
The prototype implementation is available as a part of a reversible programming system \textsc{Sparcl}, available from {{\url{https://bitbucket.org/kztk/partially-reversible-lang-impl/}}}.
Due to space limitation, we omit some proofs from this paper, which can be found in the full version~\cite{fullversion}.

 \section{Qualified-Typed Variant of $\Gl^q_\to$}
\label{sec:lang}
\label{sec:syntax}

In this section, we introduce a qualified-typed~\cite{QualifiedTypes} variant of $\Gl^q_\to$~\cite{BeBNJS18} for its rank~$1$ fragment, on which we base our type inference. 
Notable differences to the original $\Gl^q_\to$ include: (1) multiplicity abstractions and multiplicity applications
are implicit (as type abstractions and type applications), (2) this variant uses qualified typing~\cite{QualifiedTypes},
(3) conditions on multiplicities 
are inequality based~\cite{LinearMiniCore}, which gives better handling of multiplicity variables, and (4) local definitions are excluded 
as we postpone the discussions to \cref{sec:extensions} due to their issues in the handling of local assumptions in qualified typing~\cite{VytiniotisJSS11}.

\subsection{Syntax of Programs}
Programs and expressions, which will be typechecked, are given below. 
\[
\bb
\mi{prog} &\DEq& \mi{\mi{bind}}_1 ; \dots  ;\mi{bind}_n \\
\mi{bind}  &\DEq& f = e \mid f : A = e \\
e &\DEq& x \mid \lambda x.e \mid e_1 \A e_2  \mid \con{C} \A \V{e} \mid \CASE~e_0~\OF~\{ \con{C}_i \A \V{x_i} \to e_i \}_i 
\ee
\]
A program is a sequence of bindings with or without type annotations, where bound variables can appear in following bindings. As mentioned at the beginning of this section, we shall postpone the discussions of local bindings (\ie, $\LET$) to \cref{sec:extensions}.
Expressions consist of variables $x$, applications $e_1 \A e_2$, $\lambda$-abstractions $\Gl x.e$, 
constructor applications $\con{C} \A \V{e}$, and (shallow) pattern matching $\CASE~e_0~\OF~\{ \con{C}_i \A \V{x_i} \to e_i \}_i$.
For simplicity, we assume that constructors are fully-applied and patterns are shallow.
As usual, patterns $\con{C}_i \A \V{x_i}$ must be linear in the sense that each variable in $\V{x_i}$ is 
different. Programs are assumed to be appropriately $\Ga$-renamed so that newly introduced variables by $\Gl$ and patterns 
are always fresh. 
We do not require the patterns of a $\CASE$ expression to be exhaustive or no overlapping, following the original $\Gl^q_\to$~\cite{BeBNJS18};
the linearity in $\Gl^q_\to$ cares only for successful computations.
Unlike the original $\Gl^q_\to$, we do not annotate $\Gl$ and $\CASE$ with the multiplicity of the argument and the scrutinee, respectively.

Constructors play an important role in $\Gl^q_\to$. 
As we will see later, they can be used to witness unrestrictedness, similarly to $!$ of $!e$ in a linear type system~\cite{Wadler93}.

\subsection{Types}

\begin{figure}[t]
\[
\bbt
A,B    &\DEq& \forall \V{p}\V{a}. Q \To \tau  & \text{(polytypes)} \\
 \Gs, \tau &\DEq& a \mid \con{D} \A \V{\mu} \A\V{\Gt} \mid \Gs \to_\mu \tau  & \text{(monotypes)} \\
 \mu    &\DEq& p \mid 1 \mid \omega  & \text{(multiplicities)} 
\ee\quad
\bbt
Q    &\DEq& \bigwedge_i \phi_i & \text{(constraints)}\\
\phi  &\DEq& M \le M' & \text{(predicates)}\\
 M,N    &\DEq& \prod_i \mu_i & \text{(multiplications)}
\ee
\]
\caption{Types and related notions: $a$ and $p$ are type and multiplicity variables, respectively, and $\con{D}$ represents a type constructor.}
\label{fig:types}
\end{figure}

Types and related notations are defined in \cref{fig:types}. 
Types are separated into monotypes and polytypes (or, type schemes).
Monotypes consist of (rigid) type variables $a$, datatypes $\con{D} \A \V{\mu} \A \V{\Gt}$, and
multiplicity-annotated function types $\tau_1 \to_\mu \tau_2$. 
Here, a multiplicity $\mu$ is either $1$ (linear), $\omega$ (unrestricted), or a (rigid) multiplicity variable $p$. 
Polytypes have the form $\forall \V{p}\V{a}.Q \To \Gt$, where $Q$ is a constraint that is a conjunction of predicates. 
A predicate $\phi$ has the form of $M \le M'$, where $M'$ and $M$ are
multiplications of multiplicities. 
We shall sometimes treat $Q$ as a set of predicates, which means that we shall rewrite $Q$ according to contexts
by the idempotent commutative monoid laws of $\wedge$. 
We call both multiplicity ($p$) and type ($a$) variables \emph{type-level variables}, and 
write $\ftv(\V{t})$ for the set of free type-level variables in syntactic objects (such as types and constraints) $\V{t}$.

The relation $(\le)$ and operator $(\cdot)$ in predicates denote the corresponding relation and operator on $\set{1, \omega}$, respectively.
On $\set{1,\omega}$, $(\le)$ is defined as the reflexive closure of $1 \le \omega$; 
note that $(\set{1,\omega},\le)$ forms a total order. 
Multiplication $(\cdot)$ on $\set{1, \omega}$ is defined by
\[
  1 \cdot m = m \cdot 1 = m \qquad \omega \cdot m = m \cdot \omega = \omega\text{.}
\]
For simplicity, we shall sometimes omit $(\cdot)$ and write $m_1 m_2$ for $m_1 \cdot m_2$. 
Note that, for $m_1,m_2 \in \set{1,\omega}$, $m_1 \cdot m_2$ is the least upper bound of $m_1$ and $m_2$ with respect to $\le$. 
As a result, $m_1 \cdot m_2 \le m$ holds if and only if $(m_1 \le m) \wedge (m_2 \le m)$ holds;
we will use this property for efficient handling of constraints~(\cref{sec:entailment-checking}).

We assume a fixed set of constructors given beforehand. Each constructor is assigned a type of the form 
\(
 \forall \V{p}\V{a}. \: \Gt_1 \pto{\mu_1} \dots \pto{\mu_{n_1}} \Gt_n \pto{\mu_n} \con{D} \A \V{p} \A \V{a}
\)
where each $\Gt_i$ and $\mu_i$ do not contain free type-level variables other than $\{\V{p}\V{a}\}$, \ie, $\bigcup_i \ftv(\Gt_i,\mu_i) \subseteq \set{\V{p}\V{a}}$. For simplicity, 
we write the above type as $\forall \V{p}\V{a}.\: \V{\Gt} \pto{\V{\mu}} \con{D} \A \V{p} \A \V{a}$. 
We assume that types are {well-kinded}, which effectively means that $\con{D}$ is 
applied to the same numbers of multiplicity arguments and type arguments among the constructor types. 
Usually, it suffices to use constructors of linear function types as below
because they can be used in both linear and unrestricted code. 
\begin{gather*}
(-,-) : \forall a\,b.\: a \pto{1} b \pto{1} a \ox b \\
\con{Nil} : \forall a.\:\con{List} \A a \qquad \con{Cons} : \forall a.\:a \pto{1} \con{List} \A a \pto{1} \con{List} \A a 
\end{gather*}

In general, constructors can encapsulate arguments' multiplicities as below, 
which is useful when a function returns both linear and unrestricted results. 
\[
  \con{MkUn} : \forall a.\: a \pto{\omega} \con{Un} \A a \qquad 
  \con{MkMany} : \forall p\,a.\: a \pto{p} \con{Many} \A p \A a
\]
For example, a function that reads a value from a mutable array at a given index can be given as a primitive of type
$\var{readMArray} : \forall a.\: \con{MArray} \A a \pto{1} \con{Int} \pto{\omega} (\con{MArray} \A a \ox \con{Un} \A a)$~\cite{BeBNJS18}.
Multiplicity-parameterized constructors become useful when the multiplicity of contents can vary. For example, the type $\con{IO}_\mr{L} \A p \A a$ with the constructor
$\con{MkIO}_\mr{L} : (\con{World} \pto{1} (\con{World} \ox \con{Many} \A p \A a)) \pto{1} \con{IO}_\mr{L} \A p \A a$ can represent the IO monad~\cite{BeBNJS18} with methods $\var{return} : \forall p\,a.\: \var{a} \pto{p} \con{IO}_\mr{L} \A p \A a$ and $(\BIND) : \forall p\,q\,a\,b.\: \con{IO}_\mr{L} \A p \A a \pto{1} (a \pto{p} \con{IO}_\mr{L} \A q \A b) \pto{1} \con{IO}_\mr{L} \A q \A b$.

\subsection{Typing Rules}

Our type system uses two sorts of environments
A \emph{typing environment} maps variables into polytypes (as usual in non-linear calculi), 
and a \emph{multiplicity environment} maps variables into multiplications of multiplicities.
This separation of the two will be convenient when we discuss type inference. 
As usual, we write $x_1:A_1,\dots,x_n:A_n$ instead of $\set{x_1 \mapsto A_1,\dots, x_n \mapsto A_n}$ for typing environments.
For multiplicity environments, we use multiset-like notation as $\Used{x_1}{M_1},\dots, \Used{x_n}{M_n}$.

We use the following operations on multiplicity environments:\footnote{
In these definitions, we implicitly consider multiplicity $0$ and regard $\GD(x) = 0$ if $x \not\in \dom(\GD)$. 
It is natural that $0 + m = m + 0$. With $0$, multiplication $\cdot$, which is extended as $0 \cdot m = m \cdot 0 = 0$, 
no longer computes the least upper bound. Therefore, we use $\Lub$ for the last definition; in fact, the definition
corresponds to the pointwise computation of $\GD_1(x) \Lub \GD_2(x)$, where $\le$ is extended as $0 \le \omega$ but not $0 \le 1$.
This treatment of $0$ coincides with that in the Linear Haskell proposal~\cite{LinearHaskellProposal}.
}
\begin{align*}
(\GD_1 + \GD_2)(x) &= 
\begin{cases}
 \omega   & \text{if}~x \in \dom(\GD_1) \cap \dom(\GD_2) \\
 \GD_i(x) & \text{if}~x \in \dom(\GD_i) \setminus \dom(\GD_j) ~(i \ne j \in \{1,2\}) \\
\end{cases}
\\
 (\mu \GD)(x) &= \mu \cdot \GD(x)
\\
(\GD_1 \Lub \GD_2)(x) &= 
 \begin{cases}
  \GD_1(x) \cdot \GD_2(x) & \text{if}~x \in \dom(\GD_1) \cap \dom(\GD_2) \\ 
  \omega                 & \text{if}~x \in \dom(\GD_i) \setminus \dom(\GD_j)~(i \ne j \in \{1,2\})
 \end{cases}
\end{align*}
Intuitively, $\GD(x)$ represents the number of uses of $x$.
So, in the definition of $\GD_1 + \GD_2$, we have $(\GD_1 + \GD_2)(x) = \omega$ if 
$x \in \dom(\GD_1) \cap \dom(\GD_2)$ because this condition means that $x$ is used in two places. 
Operation $\GD_1 \Lub \GD_2$ is used for $\CASE$ branches. 
Suppose that a branch $e_1$ uses variables as $\GD_1$ and another branch $e_2$ uses variables as $\GD_2$.
Then, putting the branches together, variables are used as $\GD_1 \Lub \GD_2$. 
The definition says that $x$ is considered to be used linearly in the two branches put together if and only if 
both branches use $x$ linearly, where non-linear use includes unrestricted use ($\GD_i(x) = \omega$) and non-use ($x \not\in \dom(\GD)$).

We write $Q \models Q'$ if $Q$ logically entails $Q'$. That is, for any valuation of multiplicity variables $\Gth(p) \in \{1, \omega\}$, $Q'\theta$ holds if $Q \theta$ does. For example, we have $p \le r \wedge r \le q \models p \le q$. 
We extend the notation to multiplicity environments and write $Q \models \GD_1 \le \GD_2$ if $\dom(\GD_1) \subseteq \dom(\GD_2)$ and $Q \models \bigwedge_{x \in \dom(\GD)} \GD_1(x) \le \GD_2(x) \wedge \bigwedge_{x \in\dom(\GD_2)\setminus\dom(\GD_1)} \omega \le \GD_2(x)$ hold. 
We also write $Q \models \GD_1 = \GD_2$ if both $Q \models \GD_1 \le \GD_2$ and $Q \models \GD_2 \le \GD_1$ hold. 
We then have the following properties.
\begin{lemma}\rm
Suppose $Q \models \GD \le \GD'$ and $Q \models \GD = \GD_1+\GD_2$.
Then, there are some $\GD'_1$ and $\GD'_2$ such that $Q \models \GD' = \GD'_1 + \GD'_2$, $Q \models \GD_1 \le \GD'_1$ and $Q \models \GD_2 \le \GD'_2$. \qed 
\end{lemma}
\begin{lemma}\rm 
$Q \models \mu \GD \le \GD'$ implies $Q \models \GD \le \GD'$. \qed
\end{lemma}
\begin{lemma}\rm
$Q \models \GD_1 \Lub \GD_2 \le \GD'$ implies $Q \models \GD_1 \le \GD'$ and $Q \models \GD_2 \le \GD'$.\qed
\end{lemma}

Constraints $Q$ affect type equality; for example, under $Q = p \le q \wedge q \le p$, $\Gs \pto{p} \Gt$ and $\Gs \pto{q} \Gt$ become equivalent. 
Formally, we write $Q \models \Gt \eqty \Gt'$ if $\Gt\Gth = \Gt'\Gth$ for any valuation $\Gth$ of multiplicity variables that makes $Q\Gth$ true.

\begin{figure}[t]\small
\setlength{\jot}{1.4ex}
\setlength{\abovedisplayskip}{0pt}
\setlength{\belowdisplayskip}{0pt}
\begin{gather*}
\ninfer{Eq}
{
\NTy{\GG}{\GD}{e}{\Gt}{Q}
}
{
\bbc
\NTy{\GG}{\GD'}{e}{\Gt'}{Q} 
\\
Q \models \GD = \GD' \quad Q \models \Gt \eqty \Gt'
\ee
}
\qquad 
\ninfer{\rVar}
{
\NTy{\GG}{\GD}{x}{\Gt[\V{\Subst{p}{\mu}}, \V{\Subst{a}{\Gt}}]}{Q}
}
{
\bbc
\GG(x) = \forall \V{p} \V{a}. Q' \To \Gt \\
Q \models Q'[\V{\Subst{p}{\mu}}]
\quad
Q \models \Used{x}{1} \le \GD 
\ee
}
\\
\ninfer{\rAbs}
{
\NTy{\GG}{\GD}{\Gl x. e}{\Gs \pto{\mu} \Gt}{Q}
}
{
\NTy{\GG,x:\Gs}{\GD, \Used{x}{\mu}}{e}{\Gt}{Q}
}
\qquad
\ninfer{\rApp}
{
\NTy{\GG}{\GD_1 + \mu\GD_2}{e_1 \A e_2}{\Gt}{Q} 
}
{
\NTy{\GG}{\GD_1}{e_1}{\Gs \pto{\mu} \Gt}{Q}
&
\NTy{\GG}{\GD_2}{e_2}{\Gs}{Q}
}
\\
\ninfer{\rCon}
{
\NTy{\GG}{\omega\GD_0 + \sum_i \nu_i[\V{\Subst{p}{\mu}}] \GD_i}{\con{C} \A \V{e}}{\con{D} \A \V{\mu} \A \V{\Gs}}{Q}
}
{
 \con{C} : \forall \V{p} \V{a}.\: \V{\Gt} \pto{\V{\nu}} \con{D} \A \V{p} \A \V{a} &
\{ \NTy{\GG}{\GD_i}{e_i}{\Gt_i[ \V{\Subst{p}{\mu}}, \V{\Subst{a}{\Gs}}] }{Q} \}_i
}
\\
\ninfer{\rCase}
{
  \NTy{\GG}{\mu_0 \GD_0 + \bigsqcup_i \GD_i}{\CASE~e_0~\OF~\{ \con{C}_i \A \V{x_i} \to e_i \}_i}{\Gt'}{Q}
}
{
  \bb
  \NTy{\GG}{\GD_0}{e_0}{\con{D} \A \V{\mu} \A \V{\Gs}}{Q} \\\left\{
  \bb
   \con{C}_i : \forall \V{p}\V{a}. \: \V{\Gt_i} \pto{\V{\nu_i}} \con{D} \A \V{p} \A \V{a} \\
  \NTy{\GG,\V{x_i:\Gt_i[\V{\Subst{p}{\mu}},\V{\Subst{a}{\Gs}}]}}{\GD_i, \V{ \Used{x_i}{\mu_0 \nu_i[\V{\Subst{p}{\mu}}]} }}{e_i}{\Gt'}{Q}
\ee
  \right\}_i
\ee
}
\end{gather*}
\caption{Typing relation for expressions}
\label{fig:typing}
\end{figure}

\begin{figure}[t]
\small
\setlength{\jot}{1.4ex}
\setlength{\abovedisplayskip}{0pt}
\setlength{\belowdisplayskip}{0pt}
\begin{gather*}
\ninfer{\rEmpty}
{
\NTyP{ \GG }{ \Empty }
}
{}
\quad
\ninfer{\rBind}
{
\NTyP{ \GG }{ f = e; \mi{prog} } 
}
{
\NTy{\GG}{\GD}{e}{\Gt}{ Q } & \V{p}\V{a} = \ftv(Q,\Gt) &
\NTyP{ \GG, f : \forall \V{p}\V{a}. Q \To \Gt }{ \mi{prog} } 
}
\\
\ninfer{\rBindA}
{
\NTyP{\GG}{ f : (\forall \V{p}\V{a}. Q \To \Gt) = e; \mi{prog} }
}
{
\NTy{\GG}{\GD}{e}{\Gt}{Q} & \V{p}\V{a} = \ftv(Q,\Gt) & 
\NTyP{ \GG, f : \forall \V{p}\V{a}. Q \To \Gt }{ \mi{prog} } 
}
\end{gather*}
\caption{Typing rules for programs}
\label{fig:typing-prog}
\end{figure}

Now, we are ready to define the \emph{typing judgment for expressions}, $\NTy{\GG}{\GD}{e}{\Gt}{Q}$, which reads 
that under assumption $Q$, typing environment $\GG$, and multiplicity environment $\GD$, expression $e$ has monotype $\Gt$,
by the typing rules in \cref{fig:typing}. Here, we assume $\dom(\GD) \subseteq \dom(\GG)$. 
Having $x \in \dom(\GG) \setminus \dom(\GD)$ means that the multiplicity of $x$ is essentially $0$ in $e$.

Rule \rname{Eq} says that we can replace $\Gt$ and $\GD$ with equivalent ones in typing. 

Rule {\rVar} says that $x$ is used once in a variable expression $x$, but 
it is safe to regard that the expression uses $x$ more than once and uses other variables $\omega$ times. 
At the same time, the type $\forall \V{p}\V{a}. Q' \To \Gt$ of $x$ instantiated to $\Gt[\V{p \mapsto \mu}, \V{a \mapsto \Gs}]$ 
with yielding constraints $Q'[\V{p \mapsto \mu}]$, which must be entailed from $Q$.

Rule {\rAbs} says that $\Gl x.e$ has type $\Gs \pto{\mu} \Gt$
if $e$ has type $\Gt$, assuming that the use of $x$ in $e$ is $\mu$. 
Unlike the original $\Gl^q_\to$~\cite{BeBNJS18}, in our system, multiplicity annotations on arrows must be $\mu$, \ie, $1$, $\omega$, or a multiplicity variable, instead of $M$.
This does not limit the expressiveness because such general arrow types can be represented by type $\Gs \pto{p} \Gt$ with constraints $p \le M \wedge M \le p$.

Rule {\rApp} sketches an important principle in $\Gl^q_\to$; when an expression with variable use $\GD$ is used $\mu$-many times, 
the variable use in the expression becomes $\mu \GD$. 
Thus, since we pass $e_2$ (with variable use $\GD_2$) to $e_1$, where $e_1$ uses the argument $\mu$-many times as described in its type $\Gs \pto{\mu} \Gt$, 
the use of variables in $e_2$ of $e_1 \A e_2$ becomes $\mu\GD_2$. For example, for $(\Gl y. 42) \A x$, $x$ is considered to be used $\omega$
times because $(\Gl y. 42)$ has type $\Gs \pto{\omega} \con{Int}$ for any $\Gs$.

Rule {\rCon} is nothing but a combination of {\rVar} and {\rApp}.
The $\omega \GD_0$ part is only useful when $\con{C}$ is nullary; 
otherwise, we can weaken $\GD$ at leaves.

Rule {\rCase} is the most complicated rule in this type system. 
In this rule, $\mu_0$ represents how many times the scrutinee $e_0$ is used in the $\CASE$.
If $\mu_0 = \omega$, the pattern bound variables can be used unrestrictedly, 
and if $\mu_0 = 1$, the pattern bound variables can be used according to the multiplicities of the arguments of the constructor.\footnote{
This behavior, inherited from $\Gl^q_\to$~\cite{BeBNJS18}, 
implies the isomorphism $!(A \ox B) \equiv {!A} \ox {!B}$, 
which is not a theorem in the standard linear logic.
The isomorphism intuitively means that unrestricted products can (only) be constructed from 
unrestricted components, as commonly adopted in linearity-via-kind approaches \cite{GanTM15,Morris16,TovP11,MaZZ10,TurnerWM95}.
}
Thus, in the $i$th branch, variables in $\V{x_i}$ can be used as $\V{\mu_0 \nu_i[\VSubst{p}{\mu}]}$, where $\mu_i[\VSubst{p}{\mu}]$ represents 
the multiplicities of the arguments of the constructor $\con{C}_i$. 
Other than $\V{x_i}$, 
each branch body $e_i$ can contain free variables used as $\GD_i$. Thus, the uses of free variables in the whole branch bodies are summarized as $\bigsqcup_i \GD_i$. 
Recall that the $\CASE$ uses the scrutinee $\mu_0$ times; thus,
the whole uses of variables are estimated as $\mu_0 \GD_0 + \bigsqcup_i \GD_i$.

Then, we define the \emph{typing judgment for programs}, $\NTyP{\GG}{ \mi{prog} }$, which reads that program $\mi{prog}$ is well-typed under $\GG$, 
by the typing rules in \cref{fig:typing-prog}. 
At this place, the rules {\rBind} and {\rBindA} have no significant differences; their difference
will be clear when we discuss type inference. 
In the rules {\rBind} and {\rBindA}, we assumed that $\GG$ contains no free type-level variables. 
Therefore, we can safely generalize all free type-level variables in $Q$ and $\Gt$. 
We do not check the use $\GD$ in both rules as bound variables are assumed to be used 
arbitrarily many times in the rest of the program; that is, 
the multiplicity of a bound variable is $\omega$ and its body uses variable as $\omega \GD$,
which maps $x \in \dom(\GD)$ to $\omega$ and has no free type-level variables.

\subsection{Metatheories}
\label{sec:meta}

\Cref{lemma:weakening} is the standard weakening property.
\Cref{lemma:weakening-constraint} says that we can replace $Q$ with a stronger one, 
\cref{lemma:weakening-multiplicity} says that we can replace $\GD$ with a greater one, 
and \cref{lemma:substitution-type} says that we can substitute type-level variables in a term-in-context
without violating typeability. 
These lemmas state some sort of weakening, and the last three lemmas 
clarify the goal of our inference system discussed in \cref{sec:inference}. 

\begin{lemma}\rm
$\NTy{\GG}{\GD}{e}{\Gt}{Q}$ implies $\NTy{\GG,x:A}{\GD}{e}{\Gt}{Q}$. \qed
\label{lemma:weakening}
\end{lemma}
\begin{lemma}\rm
$\NTy{\GG}{\GD}{e}{\Gt}{Q}$ and $Q' \models Q$ implies $\NTy{\GG}{\GD}{e}{\Gt}{Q'}$.\qed
\label{lemma:weakening-constraint}
\end{lemma}
\begin{lemma}\rm
$\NTy{\GG}{\GD}{e}{\Gt}{Q}$ and $Q \models \GD \le \GD'$ implies $\NTy{\GG}{\GD'}{e}{\Gt}{Q}$.\qed
\label{lemma:weakening-multiplicity}
\end{lemma}
\begin{lemma}\rm
$\NTy{\GG}{\GD}{e}{\Gt}{Q}$ implies $\NTy{\GG\Gth}{\GD\Gth}{e}{\Gt\Gth}{Q\Gth}$.\qed
\label{lemma:substitution-type}
\end{lemma}

We have the following form of the substitution lemma:
\begin{lemma}[Substitution]\rm 
Suppose $\NTy{\GG,\V{x:\Gs}}{\GD_0,\V{\Used{x}{\mu}}}{e}{\Gt}{Q_0}$, and 
$\NTy{\GG}{\GD_i}{e_i'}{\Gs_i}{Q_i}$ for each $i$.
Then, 
$\NTy{\GG}{\GD_0 + \sum_i \mu_i \GD_i}{e[\V{\Subst{x}{e'}}]}{\Gt}{Q_1 \wedge \bigwedge_i Q_i}$. \qed 
\label{lemma:substitution}
\end{lemma}

\newcommand{\Red}[2]{#1 \longrightarrow #2}

\paragraph{Subject Reduction}
We show the subject reduction property for a simple call-by-name semantics. 
Consider the standard small-step call-by-name relation $\Red{e}{e'}$ with the following $\beta$-reduction rules (we omit the congruence rules):
\begin{gather*}
\Red{(\Gl x.e_1) \A e_2}{e_1[\Subst{x}{e_2}]}
\quad\qquad
\Red{ \CASE~\con{C}_j \A \V{e}_j~\OF~\{ \con{C}_i \A \V{x_i} \to e_i' \}_i }{ e'_j[\VSubst{x_j}{e_j}] } 
\end{gather*}

Then, 
we have the following subjection reduction property:
\begin{lemma}[Subject Reduction]\rm 
$\NTy{\GG}{\GD}{e}{\Gt}{Q}$ and $\Red{e}{e'}$
implies $\NTy{\GG}{\GD}{e'}{\Gt}{Q}$. 
\label{lemma:subject-reduction}
\end{lemma}

\iffullversion
\begin{proof}
Simple induction on the derivation, appealing to \cref{lemma:substitution}
for $\beta$-reduction rules above; we also use \cref{lemma:weakening-multiplicity} to deal with $\bigsqcup_i \GD_i$ in $\CASE$.
\qed
\end{proof}
\fi

\Cref{lemma:subject-reduction} holds even for the call-by-value reduction, though with a caveat. 
For a program $f_1 = e_1;\dots;f_n = e_n$, 
it can happen that some $e_i$ is typed only under unsatisfiable (\ie, conflicting) $Q_i$. 
As conflicting $Q_i$ means that $e_i$ is essentially ill-typed, evaluating $e_i$ may not be safe.
However, the standard call-by-value strategy evaluates $e_i$, even when $f_i$ is not used at all and 
thus the type system does not reject this unsatisfiability. 
This issue can be addressed by the standard witness-passing transformation~\cite{QualifiedTypes}
that converts programs so that $Q \To \Gt$ becomes $W_Q \to \Gt$, where $W_Q$ represents a set of witnesses of $Q$.
Nevertheless, it would be reasonable to reject conflicting constraints locally.

We then state the correspondence with the original system~\cite{BeBNJS18} (assuming the 
modification~\cite{LinearMiniCore} for the variable case\footnote{
In the premise of \rVar, the original~\cite{BeBNJS18} uses $\exists \GD'.\: \GD = \Used{x}{1} + \omega \GD'$, 
which is modified to $\Used{x}{1} \le \GD$ in \cite{LinearMiniCore}.
The difference between the two becomes clear when $\GD(x) = p$, for which the former one does not hold as we are not able to choose $\GD'$ 
depending on $p$.}) to show 
that the qualified-typed version captures the linearity as the original. While 
the original system assumes the call-by-need evaluation, \cref{lemma:subject-reduction} could be lifted to that case.
\begin{theorem}\rm
If $\NTy{\GG}{\GD}{e}{\Gt}{\top}$ where $\GG$ contains only monotypes, $e$ is also well-typed in the original $\Gl^q_\to$ under some environment. \qed 
\label{thm:correspondence-with-original}
\end{theorem}
\iffullversion
\begin{proof}
See \cref{proof:correspondence-with-original}. \qed
\end{proof}
\fi
The main reason for the monotype restriction is that our 
polytypes are strictly more expressive than their (rank-1) polytypes.
This extra expressiveness comes from predicates of the form $\dots \le M \cdot M'$. 
Indeed, 
$f = \Gl x.\CASE~x~\OF~\{ \con{MkMany} \A y \to (y,y) \}$ has type $\forall p\,q\,a.\: \omega \le p \cdot q \To \con{MkMany} \A p \A a \pto{q} a \ox a$ in our system, while it has three incomparable types in the original $\Gl^q_\to$.

  \section{Type Inference}
\label{sec:inference}

In this section, we give a type inference method for the type system in
the previous section. Following \cite[Section 3]{VytiniotisJSS11}, we adopt the standard two-phase 
approach; we first gather constraints on types and then solve them. 
As mentioned in \cref{sec:intro}, the inference system described here has the issue of ambiguity, 
which will be addressed in \cref{sec:defaulting}.

\subsection{Inference Algorithm}

We first extend types $\Gt$ and multiplicities $\mu$ to include \emph{unification variables}. 
\[
 \Gt \DEq \dots \mid \Ga \qquad \mu \DEq \dots \mid \pi 
\]
We call $\Ga$/$\Gp$ a unification type/multiplicity variable, which will be substituted by a concrete type/multiplicity (including rigid variables) during the inference.  
Similarly to $\ftv(\V{t})$, we write $\fuv(\V{t})$ for the unification variables (of both sorts) in $\V{t}$, where each ${t_i}$ ranges over any syntactic element (such as $\Gt$, $Q$, $\GG$, and $\GD$).

Besides $Q$, the algorithm will generate equality constraints $\Gt \eqty \Gt'$.
Formally, the sets of \emph{generated constraints} $C$ and \emph{generated predicates} $\psi$ are given by
\[
  C \DEq \bigwedge_i \psi_i \qquad \psi \DEq \phi \mid \Gt \eqty \Gt' \]

\begin{figure}[t]\small
\setlength{\jot}{1.4ex}
\setlength{\abovedisplayskip}{0pt}
\setlength{\belowdisplayskip}{0pt}
\begin{gather*}
\infer
{
 \ITy{\GG}{\Used{x}{1}}{x}{\Gt[\V{\Subst{p}{\pi}}, \V{\Subst{a}{\Ga}}]}{Q[\V{\Subst{p}{\pi}}]}
}
{
\GG(x) = \forall \V{p} \V{a}. Q \To \Gt & \V{\Ga},\V{\pi}: \text{fresh}
}
\quad
\infer
{
 \ITy{\GG}{\GD}{\Gl x.e}{\Ga \pto{\pi} \Gt}{C \wedge  M \le \pi}
}
{
 \ITy{\GG,x:\Ga}{\GD,\Used{x}{M}}{e}{\Gt}{C} & \Ga,\pi: \text{fresh}
}
\\
\infer
{
  \ITy{\GG}{\GD_1 + \pi \GD_2}{e_1 \A e_2}{\beta}{C_1 \wedge C_2 \wedge \Gt_1 \eqty (\Gt_2 \pto{\pi} \beta)}
}
{
  \ITy{\GG}{\GD_1}{e_1}{\Gt_1}{C_2}
  &
  \ITy{\GG}{\GD_2}{e_2}{\Gt_2}{C_1}
  &
  \Gb, \pi: \text{fresh}
}
\\
\infer
{
  \ITy{\GG}{\sum_i\nu_i[\V{\Subst{p}{\pi}}]\GD_i}{\con{C} \A \V{e}}{\con{D} \A \V{\pi} \A \V{\Ga}}{\bigwedge_i C_i \wedge \Gt_i \eqty \Gs_i[\VSubst{p}{\pi},\VSubst{a}{\Ga}]}
}
{
  \con{C} : \forall \V{p} \V{a}.\: \V{\Gs} \pto{\V{\nu}} \con{D} \A \V{p} \A \V{a} 
  &
  \{ \ITy{\GG}{\GD_i}{e_i}{\Gt_i}{C_i} \}_i 
  & \V{\Ga},\V{\pi}: \text{fresh}
}
\\
\infer
{
  \ITy{\GG}{\pi_0 \GD_0 + \bigsqcup_i \GD_i}{\CASE~e_0~\OF~\{ \con{C}_i \A \V{x_i} \to e_i \}_i}{\Gb}{C'}
}
{
  \bb
  \ITy{\GG}{\GD_0}{e_0}{\Gt_0}{C_0} \quad \pi_0,\V{\pi_i}, \V{\Ga_i}, \Gb: \text{fresh}
  \\
  \left\{
  \bb
  \con{C}_i : \forall \V{p}\V{a}. \: \V{\Gt_i} \pto{\V{\nu_i}} \con{D} \A \V{p} \A \V{a} 
  \\
  \ITy{\GG,\V{x_i:\Gt_i[\Subst{p}{\pi_i},\Subst{a}{\Ga_i}]}}{\GD_i, \V{\Used{x_i}{M_i}}}{e_i}{\Gt'_i}{C_i}
  \ee
  \right\}_i
  \\
  C' = C_0 \wedge \bigwedge_i \left( C_i \wedge \Gb \eqty \Gt_i' \wedge (\Gt_0 \eqty \con{D} \A \V{\pi_i} \A \V{\Ga_i}) \wedge \bigwedge_j M_{ij} \le \pi_0 \nu_{ij}[\Subst{p}{\pi_i}]  \right)
  \ee
}
\end{gather*}
\caption{Type inference rules for expressions}
\label{fig:inference-exp}
\end{figure}

\begin{figure}[t]\small
\setlength{\jot}{1.4ex}
\setlength{\abovedisplayskip}{0pt}
\setlength{\belowdisplayskip}{0pt}
\begin{gather*}
\infer
{
 \ITyP{\GG}{\Empty}
}
{
}
\qquad
\infer
{
  \ITyP{\GG}{f = e; \mi{prog}}
}
{
  \bbc
  \ITy{\GG}{\GD}{e}{\Gt}{C}
  \quad
  \ISimpl{\top}{C}{Q}{\Gth}
  \quad
  \set{\V{\pi}\V{\Ga}} = \fuv(Q, \Gt\Gth)
  \\
  \V{p},\V{a}: \text{fresh}
  \quad
  \ITyP{\GG, f : \forall \V{p}\V{a}. (Q \To \Gt\Gth)[\V{\Subst{\Ga}{a},\Subst{\pi}{p}}]}{\mi{prog}}
  \ee
}
\\
\infer
{
  \ITyP{\GG}{f : (\forall\V{p}\V{a}. Q \To \Gt) = e; \mi{prog}} 
}
{
  \bbc
  \ITy{\GG}{\GD}{e}{\Gs}{C}
  \quad
  \ISimpl{Q}{C \wedge \Gt \eqty \Gs}{\top}{\Gth}
  \quad
  \ITyP{\GG, f : \forall\V{p}\V{a}. Q \To \Gt}{\mi{prog}}
  \ee
}
\end{gather*}
\caption{Type inference rules for programs}
\label{fig:inference-prog}
\end{figure}

Then, we define \emph{type inference judgment for expressions}, $\ITy{\GG}{\GD}{e}{\Gt}{C}$, 
which reads that, given $\GG$ and $e$, type $\Gt$
is inferred together with variable use $\GD$ and constraints $C$, 
by the rules in \cref{fig:inference-exp}.
Note that $\GD$ is
also synthesized as well as $\Gt$ and $C$ in this step.
This difference in the treatment of $\GG$ and $\GD$ is why we separate multiplicity environments $\GD$ from typing environments $\GG$.

\begin{figure}[t]\small
\setlength{\jot}{1.4ex}
\setlength{\abovedisplayskip}{0pt}
\setlength{\belowdisplayskip}{0pt}
\begin{gather*}
\ninfer{\rname{S-Fun}}
{
\Simpl{Q}{(\Gs \pto{\mu} \Gt) \eqty (\Gs' \pto{\mu'} \Gt') \wedge C}{Q'}{\Gth}
}
{
\Simpl{Q}{\Gs \eqty \Gs' \wedge \mu \le \mu' \wedge \mu' \le \mu \wedge \Gt \eqty \Gt'}{Q'}{\Gth}
}
\\
\ninfer{\rname{S-Data}}
{
\Simpl{Q}{(\con{D} \A \V{\mu} \A \V{\Gs}) \eqty (\con{D} \A \V{\mu'} \A \V{\Gs'}) \wedge C}{Q'}{\Gth}
}
{
\Simpl{Q}{\V{\mu \le \mu'} \wedge \V{\mu' \le \mu} \wedge \V{\Gs \eqty \Gs'} \wedge C}{Q'}{\Gth}
}
\\
\ninfer{\rUnify}
{
\Simpl{Q}{\Ga \eqty \Gt \wedge C}{Q'}{\Gth \circ [\Subst{\Ga}{\Gt}]}
}
{
\Ga \not\in \fuv(\Gt) & \Simpl{Q}{C[\Subst{\Ga}{\Gt}]}{Q'}{\Gth}
}
\quad
\ninfer{\rname{S-Triv}}
{
\Simpl{Q}{\Gt \eqty \Gt \wedge C}{Q'}{\Gth}
}
{
\Simpl{Q}{C}{Q'}{\Gth}
}
\\
\ninfer{\rEntail}
{
\ISimpl{Q}{\phi \wedge Q_\mr{w} \wedge C}{Q'}{\Gth}
}
{
Q \wedge Q_\mr{w} \models \phi
&
\ISimpl{Q}{Q_\mr{w} \wedge C}{Q'}{\Gth}
}
\qquad
\ninfer{\rRem}
{
\ISimpl{Q}{Q'}{Q'}{\emptyset}
}
{
\text{no other rules can apply}
}
\end{gather*}
\caption{Simplification rules (modulo commutativity and associativity of $\wedge$ and commutativity of $\eqty$)}
\label{fig:simplification}
\end{figure}

Gathered constraints are solved when we process top-level bindings.
\Cref{fig:inference-prog} defines \emph{type inference judgment for programs}, $\ITyP{\GG}{\var{prog}}$, which reads that the inference finds $\var{prog}$ well-typed under $\GG$.
In the rules, manipulation of constraints is done by the \emph{simplification judgment} 
$\ISimpl{Q}{C}{Q'}{\Gth}$, 
which simplifies $C$ under the assumption $Q$ into the pair $(Q',\Gth)$ of residual constraints $Q'$ and substitution $\Gth$ for unification variables, 
where $(Q',\Gth)$ is expected to be equivalent in some sense to $C$ under the assumption $Q$.
The idea underlying our simplification is to solve type equality constraints in $C$ as much as possible
and then remove predicates that are implied by $Q$. 
Rules \rname{s-Fun}, \rname{s-Data}, \rname{s-Uni}, and \rname{S-Triv} are responsible for the former, 
which decompose type equality constraints and yield substitutions once either of the sides becomes a unification variable.
Rules {\rEntail} and {\rRem} are responsible for the latter, which 
remove predicates implied by $Q$ and then return the residual constraints.  
Rule {\rEntail} checks $Q \models \phi$; a concrete method for this check will be discussed in \cref{sec:entailment-checking}.

\begin{example}[$\var{app}$]
\label{example:app}
Let us illustrate how the system infers a type for $\var{app} = \Gl f.\Gl x. f \A x$. 
We have the following derivation for its body $\Gl f. \Gl x. f \A x$:
\[
\infer
{
 \ITy{}{\emptyset}{\Gl f.\Gl x.f \A x}{\Ga_f \pto{\pi_f} \Ga_x \pto{\pi_x} \Gb}{\Ga_f \eqty (\Ga_x \pto{\pi} \Gb) \wedge \pi_x \le \pi \wedge 1 \le \pi_f}
}
{
  \infer
  { \ITy{f : \Ga_f}{\Used{f}{1}}{\Gl x.f \A x}{ \Ga_x \pto{\pi_x} \Gb }{\Ga_f \eqty (\Ga_x \pto{\pi} \Gb) \wedge \pi_x \le \pi} }
  {
    \infer
    { \ITy{f : \Ga_f, x : \Ga_x}{\Used{f}{1},\Used{x}{\pi}}{f \A x}{ \Gb }{ \Ga_f \eqty (\Ga_x \pto{\pi} \Gb) } }
    {
      \infer{\ITy{f : \Ga_f}{\Used{f}{1}}{f}{\Ga_f}{\top}}{}
      & 
      \infer{\ITy{x : \Ga_x}{\Used{x}{1}}{x}{\Ga_x}{\top}}{}
    }
  }
}
\]
The highlights in the above derivation are: 
\begin{itemize}
 \item In the last two steps, $f$ is assigned to type $\Ga_f$ and multiplicity $\pi_f$, 
       and $x$ is assigned to type $\Ga_x$ and multiplicity $\pi_x$. 
 \item Then, in the third last step, for $f \A x$, the system infers type $\Gb$ with constraint $\Ga_f \eqty (\Ga_x \pto{\pi} \Gb)$.
       At the same time, the variable use in $f \A x$ is also inferred as $\Used{f}{1},\Used{x}{\pi}$.
       Note that the use of $x$ is $\pi$ because it is passed to $f : \Ga_x \pto{\pi} \Gb$. 
 \item After that, in the last two steps again, the system yields constraints $\pi_x \le \pi$ and $1 \le \pi_f$. 
\end{itemize}
As a result, the type $\Gt = {\Ga_f \pto{\pi_f} \Ga_x \pto{\pi_x} \Gb}$ is inferred with the constraint $C = {\Ga_f \eqty (\Ga_x \pto{\pi} \Gb) \wedge \pi_x \le \pi \wedge 1 \le \pi_f}$.

Then, we try to assign a polytype to $\var{app}$ by the rules in \cref{fig:inference-exp}.
By simplification, we have $\Simpl{\top}{C}{\pi_x \le \pi}{[\Subst{\Ga_f}{(\Ga_x \pto{\pi} \Gb)}]}$.
Thus, by generalizing $\Gt[\Subst{\Ga_f}{(\Ga_x \pto{\pi} \Gb)}] = (\Ga_x \pto{\pi} \Gb) \pto{\pi_f} \Ga_x \pto{\pi_x} \Gb$ with $\pi_x \le \pi$, 
we obtain the following type for $\var{app}$:
\begin{qedmath}
 \var{app} :
  \forall p\,p_f\,p_x\,a\,b.\: p \le p_x \To (a \pto{p} b) \pto{p_f} a \pto{p_x} b
\end{qedmath}
\end{example}

\paragraph{Correctness}
We first prepare some definitions for the correctness discussions. 
First, we allow substitutions 
$\Gth$ to replace unification multiplicity variables as well as unification type variables. 
Then, we extend the notion of $\models$ and write $C \models C'$ 
if $C' \Gth$ holds when $C \Gth$ holds. From now on, we require that substitutions are idempotent, \ie, $\Gt \Gth \Gth = \Gt\Gth$ for any $\Gt$, which 
excludes substitutions $[\Subst{\Ga}{\con{List} \A \Ga}]$ and $[\Subst{\Ga}{\Gb},\Subst{\Gb}{\con{Int}}]$ for example.
Let us write $Q \models \Gth = \Gth'$ if $Q \models \Gt\Gth \eqty \Gt\Gth'$ for any $\Gt$. 
The restriction of a substitution $\Gth$ to a domain $X$ is written by $\Gth|_X$.

Consider a pair $(Q_\mr{g}, C_\mr{w})$, where we call $Q_\mr{g}$ and $C_\mr{w}$ given and wanted constraints, respectively. 
Then, a pair $(Q,\Gth)$ is called a (sound) \emph{solution}~\cite{VytiniotisJSS11} for the pair $(Q_\mr{g},C_\mr{w})$ if 
$Q_\mr{g} \wedge Q \models C_\mr{w}\Gth$, $\Disjoint{\dom(\Gth)}{\fuv(Q_\mr{g})}$, and 
$\Disjoint{\dom(\Gth)}{\fuv(Q)}$.
A solution is called \emph{guess-free}~\cite{VytiniotisJSS11} if it satisfies $Q_\mr{g} \wedge C_\mr{w} \models Q \wedge \bigwedge_{\pi \in \dom(\Gth)} (\pi = \Gth(\pi)) \wedge \bigwedge_{\Ga \in \dom(\Gth)} (\Ga \eqty \Gth(\Ga))$ in addition. 
Intuitively, a guess-free solution consists of necessary conditions required 
for a wanted constraint $C_\mr{w}$ to hold, assuming a given constraint $Q_\mr{g}$.
For example, for $(\top, \Ga \eqty (\Gb \pto{1} \Gb))$, $(\top, [\Subst{\Ga}{(\con{Int} \pto{1} \con{Int})},\Subst{\Gb}{\con{Int}}])$ is a solution but not guess-free. 
Very roughly speaking, being for $(Q,\Gth)$ a guess-free solution of $(Q_\mr{g}, C_\mr{w})$ means that 
$(Q,\Gth)$ is equivalent to $C_\mr{w}$ under the assumption $Q_\mr{g}$. 
There can be multiple guess-free solutions; for example, for $(\top, \pi \le 1)$, 
both $(\pi \le 1, \emptyset)$ and $(\top, [\Subst{\pi}{1}])$ are guess-free solutions.

\begin{lemma}[Soundness and Principality of Simplification]\rm
If $\Simpl{Q}{C}{Q'}{\Gth}$, $(Q',\Gth)$ is a guess-free solution for $(Q, C)$. 
\label{lemma:simplification-sound}
\end{lemma}
\iffullversion
\begin{proof} See Appendix~\ref{proof:simplification-sound}. \end{proof}
\fi

\begin{lemma}[Completeness of Simplification]\rm
If $(Q',\Gth)$ is a solution for $(Q,C)$ where $Q$ is satisfiable,
then $\Simpl{Q}{C}{Q''}{\Gth'}$ for some $Q''$ and $\Gth'$. 
\label{lemma:simplification-complete}
\end{lemma}
\iffullversion
\begin{proof} See \cref{proof:simplification-complete}. \end{proof}
\fi

\begin{theorem}[Soundness of Inference]\rm
Suppose $\ITy{\GG}{\GD}{e}{\Gt}{C}$ and there is a solution $(Q,\Gth)$ for $(\top, C)$.
Then, we have $\NTy{\GG\Gth}{\GD\Gth}{e}{\Gt\Gth}{Q}$.
\label{theorem:inference-sound}
\end{theorem}

\iffullversion
\begin{proof}
See \cref{proof:inference-sound}. \qed
\end{proof}
\fi

\begin{theorem}[Completeness and Principality of Inference]\rm
Suppose $\ITy{\GG}{\GD}{e}{\Gt}{C}$. 
Suppose also that $\NTy{\GG\Gth'}{\GD'}{e}{\Gt'}{Q'}$ for some substitution $\Gth'$ on unification variables such that 
$\dom(\Gth') \subseteq \fuv(\GG)$ and $\Disjoint{\dom(\Gth')}{\fuv(Q')}$.
Then, there exists $\Gth$ such that 
$\dom(\Gth) \setminus \dom(\Gth') \subseteq X$, 
$(Q', \Gth)$ is a solution for $(\top,C)$,
$Q' \models \Gth|_{\dom(\Gth')} = \Gth'$, 
$Q' \models \Gt \Gth \eqty \Gt'$, 
and $Q' \models \GD \Gth \le \GD'$, 
where $X$ is the set of unification variables introduced in the derivation. 
\label{thm:inference-complete}
\end{theorem}

\iffullversion
\begin{proof}
See \cref{proof:inference-complete}. \qed
\end{proof}
\fi

Note that the constraint generation $\ITy{\GG}{\GD}{e}{\Gt}{C}$ always succeeds, 
whereas the generated constraints may possibly be conflicting. 
\cref{thm:inference-complete} states that such a case cannot happen when $e$ is well-typed under the rules 
in \cref{fig:typing}.

\paragraph{Incompleteness in Typing Programs.}
It may sound contradictory to \cref{thm:inference-complete}, but the type inference is indeed incomplete for checking type-annotated bindings. Recall that the typing rule for type-annotated bindings requires that the resulting constraint after simplification must be $\top$. 
However, even when there exists a solution of the form $(\top, \Gth)$ for $(Q, C)$, 
there can be no guess-free solution of this form.
For example, $(\top, \pi \le \pi')$ has a solution $(\top, [\pi \mapsto \pi'])$, but there
are no guess-free solutions of the required form. 
Also, even though there exists a guess-free solution of the form $(\top, \Gth)$, 
the simplification may not return the solution, as guess-free solutions are not always unique. 
For example, for $(\top, \pi \le \pi' \wedge \pi' \le \pi)$, $(\top, [\pi \mapsto \pi'])$ is 
a guess-free solution, whereas we have 
$\ISimpl{\top}{\pi \le \pi' \wedge \pi'\le \pi}{\pi \le \pi' \wedge \pi' \le \pi}{\emptyset}$.
The source of the issue is that constraints on multiplicities can (also) be solved by substitutions. 

Fortunately, this issue disappears when we consider disambiguation in 
\cref{sec:defaulting}.
By disambiguation, we can eliminate constraints for internally-introduced multiplicity unification variables that are invisible from the outside. 
As a result, after processing equality constraints, 
we essentially need only consider rigid multiplicity variables when checking entailment for annotated top-level bindings.

\paragraph{Promoting Equalities to Substituions.}
The inference can infer polytypes 
$\forall p.\; p \le 1 \To \con{Int} \pto{p} \con{Int}$ and  $\forall p_1\, p_2.\; (p_1 \le p_2 \wedge p_2 \le p_1) \To \con{Int} \pto{p_1} \con{Int} \pto{p_2} \con{Int}$, 
while programmers would prefer more simpler types 
$\con{Int} \pto{1} \con{Int}$ and 
$\forall p.\; \con{Int} \pto{p} \con{Int} \pto{p} \con{Int}$;
the simplification so far does not yield substitutions on multiplicity unification variables. 
Adding the following rule remedies the situation:
\[
\ninfer{\rEntailEq}
{
\ISimpl{Q}{Q_\mathrm{w} \wedge C}{Q'}{\Gth \circ [\Subst{\pi}{\mu}]}
}
{
\bbc
\pi \not\in \fuv(Q) \quad
\pi \ne \mu \\
Q \wedge Q_\mathrm{w} \models \pi \le \mu \wedge \mu \le \pi \quad
\ISimpl{Q}{(Q_\mathrm{w} \wedge C)[\Subst{\pi}{\mu}]}{Q'}{\Gth}
\ee
}
\]
This rule says that if $\pi = \mu$ must hold for $Q_\mathrm{w} \wedge C$ to hold, the simplification yields the substitution $[\Subst{\pi}{\mu}]$.
The condition $\pi \not\in \fuv(Q)$ is required for \cref{lemma:simplification-sound}; a solution cannot substitute variables in $Q$.
Note that this rule essentially finds an improving substitution~\cite{Jones95}.

Using the rule is optional. 
Our prototype implementation actually uses {\rEntailEq} only for $Q_\mr{w}$ for which we can find $\mu$ easily: $M \le 1$, $\omega \le \mu$, 
and looping chains $\mu_1 \le \mu_2 \wedge \dots \wedge \mu_{n-1} \le \mu_n \wedge \mu_n \le \mu_1$.

\subsection{Entailment Checking by Horn SAT Solving}
\label{sec:entailment-checking}
The simplification rules rely on the check of entailment $Q \models \phi$. 
For the constraints in this system, we can perform this check in quadratic time at worst but in linear time for most cases.
Specifically, we reduce the checking $Q \models \phi$ to satisfiability of propositional Horn formulas (Horn SAT), which is 
known to be solved in linear time in the number of occurrences of literals~\cite{DowlingG84}, 
where the reduction (precisely, the preprocessing of the reduction) may increase the problem size quadratically. 
The idea of using Horn SAT for constraint solving in linear typing can be found in Mogensen \cite{Mogensen97}.

First, as a preprocess, we normalize both given and wanted constraints by the following rules:
\begin{itemize}
 \item Replace $M_1 \cdot M_2 \le M$ with $M_1 \le M \wedge M_2 \le M$. 
 \item Replace $M \cdot 1$ and $1 \cdot M$ with $M$, and $M \cdot \omega$ and $\omega \cdot M$ with $\omega$.
 \item Remove trivial predicates $1 \le M$ and $M \le \omega$.
\end{itemize}
After this, each predicate $\phi$ has the form 
\(
 \mu \le \prod_i \nu_i\text{.}
\)

After the normalization above, 
we can reduce the entailment checking to satisfiability.
Specifically, we use the following property:
\[
Q \models \mu \le \prod_i \nu_i 
\quad \text{iff} \quad 
Q \wedge \bigwedge_i (\nu_i \le 1) \wedge (\omega \le \mu) \text{~is unsatisfiable}
\]
Here, the constraint $Q \wedge \bigwedge_i (\nu_i \le 1) \wedge (\omega \le \mu)$ intuitively asserts
that there exists a counterexample of $Q \models \mu \le \prod_i \nu_i$.

Then, it is straightforward to reduce the satisfiability of $Q$ to Horn SAT;
we just map $1$ to true and $\omega$ to false
and accordingly map $\le$ and $\cdot$ to $\From$ and $\wedge$, respectively.
Since Horn SAT can be solved in linear time in the number
of occurrences of literals~\cite{DowlingG84}, the reduction also shows
that the satisfiability of $Q$ is checked in linear time in the size of $Q$ if $Q$ is normalized.

\begin{corollary}\rm
Checking $Q \models \phi$ is in linear time if $Q$ and $\phi$ are normalized. \qed
\end{corollary}

The normalization of constraints can duplicate $M$ of $\dots \le M$, 
and thus increases the size quadratically in the worst case. 
Fortunately, 
the quadratic increase is not common because the size of $M$ is bounded in practice, in many cases by one.
Among the rules in \cref{fig:typing}, only the rule that introduces non-singleton $M$ in the right-hand side of $\le$ is {\rCase} for a constructor 
whose arguments' multiplicities are non-constants, such as 
$\con{MkMany} : \forall p\,a.\: a \pto{p} \con{Many} \A p \A a$. 
However, it often suffices to use non-multiplicity-parameterized constructors, such as $\con{Cons} : \forall a.\: a \pto{1} \con{List} \A a \pto{1} \con{List} \A a$,
because such constructors can be used to construct or deconstruct both linear and unrestricted data.

\subsection{Issue: Inference of Ambiguous Types}
\label{sec:ambiguity}
The inference system so far looks nice; the system is sound and complete, and infers principal types. 
However, there still exists an issue to overcome for the system to be useful: 
it often infers ambiguous types~\cite{QualifiedTypes,StuckeyS05} in which internal multiplicity variables leak out to reveal 
internal implementation details. 

Consider $\var{app}' = \Gl f.\Gl x. \var{app} \A f \A x$ for $\var{app} = \Gl f. \Gl x. f \A x$ from \cref{example:app}.
We would expect that equivalent types are inferred for $\var{app}'$ and $\var{app}$. However, this is not the case for the inference system.
In fact, the system infers the following type for $\var{app}'$ (here we reproduce the inferred type of $\var{app}$ for comparison):
\[
\begin{tarray}{lcrll}
 \var{app}  &:& \forall p\, p_f \, p_x \, a \, b.\: & (p \le p_x) & \To (a \pto{p} b) \pto{p_f} a \pto{p_x} b \\
 \var{app}' &:& \forall q\,q_f\,q_x\,p_f\,p_x\,a\,b.\: & (q \le q_x \wedge q_f \le p_f \wedge q_x \le p_x) &\To (a \pto{q} b) \pto{p_f} a \pto{p_x} b
\end{tarray}
\]
We highlight why this type is inferred as follows. \begin{itemize}
 \item By abstractions, $f$ is assigned to type $\Ga_f$ and multiplicity $\pi_f$, 
       and $x$ is assigned to type $\Ga_x$ and multiplicity $\pi_x$. 
 \item By its use, $\var{app}$ is instantiated to type $(\Ga' \pto{\pi'} \Gb') \pto{\pi'_f} \Ga' \pto{\pi'_x} \Gb'$
       with constraint $\pi' \le \pi'_x$. 
 \item For $\var{app} \A f$, the system infers type $\Gb$ with constraint
       $((\Ga' \pto{\pi'} \Gb') \pto{\pi'_f} \Ga' \pto{\pi'_x} \Gb') \eqty (\Ga_f \pto{\pi_1} \Gb)$.
       At the same time, the variable use in the expression is inferred as $\Used{\var{app}}{1},\Used{f}{\pi_1}$.
 \item For $(\var{app} \A f \A x)$, the system infers type $\Gg$ with constraint 
       $\Gb \eqty (\Ga' \pto{\pi_2} \Gg)$.
       At the same time, the variable use in the expression is inferred as $\Used{\var{app}}{1},\Used{f}{\pi_1},\Used{x}{\pi_2}$.
 \item As a result, $\Gl f. \Gl x. \var{app} \A f \A x$ has type $\Ga_f \pto{\pi_f} \Ga_x \pto{\pi_x} \Gg$,
       yielding constraints $\pi_1 \le \pi_f \wedge \pi_2 \le \pi_x$.
\end{itemize}
Then, for the gathered constraints, by simplification (including \rEntailEq), 
we obtain a (guess-free) solution $(Q,\Gth)$ such that 
\(Q = (\pi'_f \le \pi_f \wedge \pi' \le \pi'_x \wedge \pi'_x \le \pi_x)\) and
\(
\Gth = [\Subst{\Ga_f}{(\Ga' \pto{\pi'} \Gb')},\Subst{\pi'_1}{\pi'_f},\Subst{\Gb}{(\Ga_f \pto{\pi'_x} \Gb')},\Subst{\pi_2}{\pi'_x}, \Subst{\Gg}{\Gb'}])
\). Then, after generalizing $(\Ga_f \pto{\pi_f} \Ga_x \pto{\pi_x} \Gg)\Gth = (\Ga' \pto{\pi'} \Gb') \pto{\pi_f} \Ga' \pto{\pi_x} \Gb$, 
we obtain the inferred type above.

There are two problems with this inference result:
\begin{itemize} 
 \item The type of $\var{app}'$ is \emph{ambiguous} in the sense that the type-level variables in the constraint cannot be determined only by those that appear in the type~\cite{QualifiedTypes,StuckeyS05}. Usually, ambiguous types are undesirable, especially when their instantiation affects runtime behavior~\cite{QualifiedTypes,VytiniotisJSS11,StuckeyS05}. 
\item Due to this ambiguity, the types of $\var{app}$ and $\var{app}'$  are not judged equivalent by the inference system. 
   For example, the inference rejects 
   the binding $\var{app}'' : \forall p\, p_f \, p_x \, a \, b.\: (p \le p_x) \To (a \pto{p} b) \pto{p_f} a \pto{p_x} b = \var{app}'$ 
   because the system does not know how to instantiate the ambiguous type-level variables $q_f$ and $q_x$, 
   while the binding is valid in the type system in \cref{sec:lang}. 
\end{itemize}
Inference of ambiguous types is common in the system; it is easily caused by using defined variables.
Rejecting ambiguous types is not a solution for our case because it rejects many programs. 
Defaulting such ambiguous type-level variables to $1$ or $\omega$ is not a solution either because it loses principality in general.
However, we have no other choices than to reject ambiguous types, \emph{as long as multiplicities are relevant in runtime behavior}.

In the next section, we will show how we address the ambiguity issue under the assumption that 
multiplicities are irrelevant at runtime. 
Under this assumption, it is no problem to have multiplicity-monomorphic primitives such as array processing primitives (\eg, $\var{readMArray} : \forall a.\: \con{MArray} \A a \pto{1} \con{Int} \pto{\omega} (\con{MArray} \A a \ox \con{Un} \A a)$)~\cite{VytiniotisJSS11}. 
Note that this assumption does not rule out all multiplicity-polymorphic primitives; 
it just prohibits the primitives from inspecting multiplicities at runtime.

 \section{Disambiguation by Quantifier Elimination}
\label{sec:defaulting}

In this section, we address the issue of ambiguous and leaky types by using quantifier elimination. 
The basic idea is simple; we just view the type of $\var{app}'$ as
\[
\var{app}' : \forall q\,p_f\,p_x\,a\,b.\: (\exists q_x\,q_f.\: q \le q_x \wedge q_f \le p_f \wedge q_x \le p_x) \To (a \pto{q} b) \pto{p_f} a \pto{p_x} b
\]
In this case, the constraint $(\exists q_x\,q_f.\: q \le q_x \wedge q_f \le p_f \wedge q_x \le p_x)$ is logically equivalent to $q \le p_x$, 
and thus we can infer the equivalent types for both $\var{app}$ and $\var{app}'$.
Fortunately, such quantifier elimination is always possible for our representation of constraints; that is, 
for $\exists p.Q$, there always exists $Q'$ that is logically equivalent to $\exists p.Q$.
A technical subtlety is that, although we perform quantifier elimination after generalization in the above explanation, we actually 
perform quantifier elimination just before generalization, or more precisely, as a final step of simplification, 
for compatibility with the simplification in \textsc{OutsideIn(X)}~\cite{VytiniotisJSS11}, especially in the treatment of local assumptions.

\subsection{Elimination of Existential Quantifiers}

The elimination of existential quantifiers is rather easy;
we simply use the well-known fact that a disjunction of a Horn clause and a definite clause can also be represented as a Horn clause. 
Regarding our encoding of normalized predicates (\cref{sec:entailment-checking}) that maps $\mu \le M$ to a Horn clause, the fact can be rephrased as:
\begin{lemma}\rm
$(\mu \le M \vee \omega \le M') \equiv \mu \le M \cdot M'$. \label{lemma:or}\qed
\end{lemma}
Here, we extend constraints to include $\vee$ and 
write $\equiv$ for the logical equivalence; that is, $Q \equiv Q'$ if and only if $Q \models Q'$ and $Q' \models Q$.

As a corollary, we obtain the following result:
\begin{corollary}\rm 
There effectively exists a quantifier-free constraint $Q'$, denoted by $\Elim{\pi}{Q}$, such that $Q'$ is logically equivalent to $\exists \pi.Q$. 
\end{corollary}
\begin{proof}
Note that $\exists \pi. Q$ means $Q[\Subst{\pi}{1}] \vee Q[\Subst{\pi}{\omega}]$
because $\Gp$ ranges over $\set{1, \omega}$. 
We safely assume that $Q$ is normalized (\cref{sec:entailment-checking}) and that $Q$ does not contain a predicate $\pi \le M$ where $\pi$ appears also in $M$, because such a predicate 
trivially holds. 

We define $\Phi_1$, $\Phi_\omega$, and $Q_\mr{rest}$ as 
\(\Phi_1 = \set{ \mu \le M \mid (\mu \le \pi \cdot M) \in Q, \mu \ne \pi }\), 
\(\Phi_\omega = \set{ \omega \le M \mid (\Gp \le M) \in Q, \pi \not\in \fuv(M) }\), and 
\(Q_\mr{rest} = \bigwedge \set{ \phi \mid \phi \in Q, \pi \not\in\fuv(\phi) }\).
Here, we abused the notation to write $\phi \in Q$ to mean that $Q = \bigwedge_i \phi_i$ and $\phi = \phi_i$ for some $i$.
In the construction of $\Phi_1$, we assumed the monoid laws of $(\cdot)$;
the definition says that we remove $\pi$ from the right-hand sides and $M$ becomes $1$ if the right-hand side is $\pi$. 
By construction, $Q[\Subst{p}{1}]$ and $Q[\Subst{p}{\omega}]$ 
are equivalent to $(\bigwedge \Phi_1) \wedge Q_\mr{rest}$ and 
$(\bigwedge \Phi_\omega) \wedge Q_\mr{rest}$, respectively.
\iffullversion
Thus, by \cref{lemma:or} and by the distributivity of $\vee$ over $\wedge$,
it suffices to define $Q'$ as:
\begin{qedmath}
Q' = 
 \left( \bigwedge \set{ \mu \le M \cdot M' \mid \mu \le M \in \Phi_1, \omega \le M' \in \Phi_\omega} \right)
 \wedge Q_\mr{rest}
\end{qedmath}
\fi
\end{proof}

\begin{example}\label{ex:qe}
Consider 
\(Q = (\pi'_f \le \pi_f \wedge \pi' \le \pi'_x \wedge \pi'_x \le \pi_x)\); this is the constraint obtained from $\Gl f. \Gl x. \var{app} \A f \A x$ (\cref{sec:ambiguity}).
Since $\pi'_f$ and $\pi'_x$ do not appear in the inferred type $(\Ga' \pto{\pi'} \Gb') \pto{\pi_f} \Ga' \pto{\pi_x} \Gb$, 
we want to eliminate them by the above step. 
There is a freedom to choose which variable is eliminated first. Here, we shall choose $\pi'_f$ first.

First, we have $\Elim{\pi'_f}{Q} = \pi' \le \pi'_x \wedge \pi'_x \le \pi_x$ because 
for this case we have $\Phi_1 = \emptyset$, $\Phi_\omega = \set{\omega \le \pi_f}$, and $Q_\mr{rest} = \pi' \le \pi'_x \wedge \pi'_x \le \pi_x$. 
We then have $\Elim{\pi'_x}{\pi' \le \pi'_x \wedge \pi'_x \le \pi_x} = \pi' \le \pi_x$ because
for this case we have $\Phi_1 = \set{\pi' \le 1}$, $\Phi_2 = \set{\omega \le \pi_x}$, and $Q_\mr{rest} = \top$. \qed 
\end{example}

In the worst case, the size of $\Elim{\pi}{Q}$ can be quadratic to that of $Q$. Thus, repeating elimination can make the constraints 
exponentially bigger. We believe that such blow-up rarely happens because it is usual that $\pi$ occurs only in a few predicates in $Q$. 
Also, recall that non-singleton right-hand sides are caused only by multiplicity-parameterized constructors. 
When each right-hand side of $\le$ is a singleton in $Q$, the same holds in $\Elim{\pi}{Q}$.
For such a case, the exponential blow-up cannot happen because the size of constraints in the form is at most quadratic in the number of multiplicity variables. 

\subsection{Modified Typing Rules}
As mentioned at the begging of this section, 
we perform quantifier elimination as the last step of simplification. 
To do so, we define $\SimplE{Q}{C}{Q''}{\Gth}{\Gt}$ as follows:
\[
\infer
{
\SimplE{Q}{C}{Q''}{\Gth}{\Gt}
}
{
\Simpl{Q}{C}{Q'}{\Gth} 
&
\set{\V{\pi}} = \fuv(Q') \setminus \fuv(\Gt\Gth)
&
Q'' = \Elim{\V{\pi}}{Q'}
}
\]
Here, $\Gt$ is used to determine which unification variables will be ambiguous after generalization. 
We simply identify variables ($\V{\pi}$ above) that are not in $\Gt$ as ambiguous~\cite{QualifiedTypes} for simplicity.
This check is indeed conservative in a more general definition of ambiguity~\cite{StuckeyS05}, in which
$\forall p\,r\,a.\: (p \le r, r \le p) \To a \pto{p} a$ for example is not judged as ambiguous because $r$ is determined by $p$.

Then, we replace the original simplification with the above-defined version. 
\begin{gather*}
\infer
{
  \ITyP{\GG}{f = e; \mi{prog}}
}
{
  \bbc
  \ITy{\GG}{\GD}{e}{\Gt}{C}
  \quad
  \mathhighlight{\SimplE{\top}{C}{Q}{\Gth}{\Gt}}
  \quad
  \set{\V{\pi}\V{\Ga}} = \fuv(Q, \Gt\Gth)
  \\
  \V{p},\V{a}: \text{fresh}
  \quad
  \ITyP{\GG, f : \forall \V{p}\V{a}. (Q \To \Gt\Gth)[\V{\Subst{\Ga}{a},\Subst{\pi}{p}}]}{\mi{prog}}
  \ee
}
\\
\infer
{
  \ITyP{\GG}{f : (\forall\V{p}\V{a}. Q \To \Gt) = e; \mi{prog}} 
}
{
  \bbc
  \ITy{\GG}{\GD}{e}{\Gs}{C}
  \quad
  \mathhighlight{\SimplE{Q}{C \wedge \Gt \eqty \Gs}{\top}{\Gth}{\Gs}}
  \quad
  \ITyP{\GG, f : \forall\V{p}\V{a}. Q \To \Gt}{\mi{prog}}
  \ee
}
\end{gather*}
Here, the changed parts are $\mathhighlight{\text{highlighted}}$ for readability.

\begin{example}
Consider $(Q,\Gth)$ in \cref{sec:ambiguity} such that \(Q = (\pi'_f \le \pi_f \wedge \pi' \le \pi'_x \wedge \pi'_x \le \pi_x)\) and 
\(
\Gth = [\Subst{\Ga_f}{(\Ga' \pto{\pi'} \Gb')},\Subst{\pi'_1}{\pi'_f},\Subst{\Gb}{(\Ga_f \pto{\pi'_x} \Gb')},\Subst{\pi_2}{\pi'_x}, \Subst{\Gg}{\Gb'}])
\), which is obtained after simplification of the gathered constraint. 
Following \cref{ex:qe}, eliminating variables that are not in $\Gt\Gth = (\Ga' \pto{\pi'} \Gb') \pto{\pi_f} \Ga' \pto{\pi_x} \Gb$
yields the constraint $\pi' \le \pi_x$. 
As a result, by generalization, we obtain the polytype 
\[
 \forall q\,p_f\,p_x\,a\,b.\: (q \le p_x) \To (a \pto{q} b) \pto{p_f} a \pto{p_x} b
\]
for $\var{app}'$, which is equivalent to the inferred type of $\var{app}$. \qed
\end{example}

Note that $(Q',\Gth)$ of $\SimplE{Q}{C}{Q'}{\Gth}{\Gt}$ is no longer a solution of $(Q,C)$ 
because $C$ can have eliminated variables. 
However, it is safe to use this version when generalization takes place, because, 
for variables $\V{q}$ that do not occur in $\Gt$, 
$\forall \V{p}\V{q}\V{a}.\:Q \To \Gt$ and $\forall \V{p}\V{a}.\: Q' \To \Gt$ have the 
same set of monomorphic instances, if $\exists \V{q}.Q$ is logically equivalent to $Q'$. 
Note that in this type system simplification happens only before (implicit) generalization takes place.

 \section{Extension to Local Assumptions}
\label{sec:extensions}

In this section, following \textsc{OutsideIn(X)}~\cite{VytiniotisJSS11}, 
we extend our system with local assumptions, which enable us to have $\LET$s and GADTs. 
We focus on the treatment of $\LET$s in this section
because type inference for $\LET$s involves a linearity-specific concern: the multiplicity of a $\LET$-bound variable.

\subsection{``Let Should Not Be Generalized'' for Our Case}
We first discuss that even for our case ``$\LET$ should not be generalized''~\cite{VytiniotisJSS11}. 
That is, generalization of $\LET$ sometimes results in counter-intuitive typing and conflicts with the discussions so far. 

Consider the following program:
\[
\var{h} = \Gl f. \Gl k. \LET \A y = f \A (\Gl x.k \A x)~\IN~0 
\]
Suppose for simplicity that $f$ and $x$ have types 
$(a \pto{\Gp_1} b) \pto{\Gp_2} c$ and $a \pto{\Gp_3} b$, respectively (here we only focus on the treatment of multiplicity).
Then, $f \A (\Gl x.k \A x)$ has type $c$ with the constraint $\Gp_3 \le \Gp_1$. 
Thus, after generalization, $y$ has type $\Gp_3 \le \Gp_1 \To c$, where $\Gp_3$ and $\Gp_1$ are neither generalized nor 
eliminated because they escape from the definition of $y$. 
As a result, $h$ has type $\forall p_1\,p_2\,p_3\,a\,b\,c.\:((a \pto{p_1} b) \pto{p_2} c) \pto{\omega} (a \pto{p_3} b) \pto{\omega} \con{Int}$; 
there is no constraint $p_3 \le p_1$ because the definition of $y$ does not yield a constraint. 
This nonexistence of the constraint would be counter-intuitive because users wrote $f \A (\Gl x. k \A x)$ while the constraint for the expression is not imposed. 
In particular, it does not cause an error even when $f : (a \pto{1} b) \pto{1} c$ and $k : a \pto{\omega} b$, while $f \A (\Gl x. k \A x)$ becomes illegal 
for this case.
Also, if we change $0$ to $y$, the error happens at the use site instead of the definition site.
Moreover, the type is fragile as it depends on whether $y$ occurs or not;
for example, if we change $0$ to $\var{const} \A 0 \A y$ where $\var{const} = \Gl a.\Gl b.a$, 
the type of $h$ changes to $\forall p_1\,p_2\,p_3\,a\,b\,c.\:p_1 \le p_3 \To ((a \pto{p_1} b) \pto{p_2} c) \pto{\omega} (a \pto{p_3} b) \pto{\omega} \con{Int}$.
In this discussion, we do not consider type-equality constraints, but there are no legitimate reasons why type-equality constraints are solved on the fly in typing $y$.

As demonstrated in the above example, ``$\LET$ should not be generalized''~\cite{VytiniotisJS10,VytiniotisJSS11} in our case.
Thus, we adopt the same principle in \textsc{OutsideIn(X)} that $\LET$ will be generalized only if users write a type annotation for it~\cite{VytiniotisJSS11}.
This principle is also adopted in GHC (as of 6.12.1 when the language option \texttt{MonoLocalBinds} is turned on) with a slight relaxation to generalize closed bindings. 

\subsection{Multiplicity of Let-Bound Variables}

Another issue with $\LET$-generalization, which is specific to linear typing, 
is that a generalization result depends on the multiplicity of the $\LET$-bound variable.
Let us consider the following program, where we want to generalize the type of $y$ (even without a type annotation):
\[
\var{g} = \Gl x. \LET \A y = \Gl f. f \A x~\IN~ y \A \var{not}
\]
Suppose for simplicity that $\con{not}$ has type $\con{Bool} \pto{1} \con{Bool}$ and $x$ has type $\con{Bool}$ already in typing $\LET$. 
Then, $y$'s body $\Gl f.f \A x$ has a monotype $(\con{Bool} \pto{\Gp} r) \pto{\Gp'} r$ with no constraints (on multiplicity).
There are two generalization results depending on the multiplicity $\pi_y$ of $y$ because the use of $x$ also escapes in the type system. 
\begin{itemize}
 \item If $\pi_y = 1$, the type is generalized into $\forall q\,r.\:(\con{Bool} \pto{\Gp} r) \pto{q} r$, where $\Gp$ is not generalized because
       the use of $x$ in $y$'s body is $\Gp$.
 \item If $\pi_y = \omega$, the type is generalized into $\forall p\,q\,r.\:(\con{Bool} \pto{p} r) \pto{q} r$, where $\Gp$ is generalized (to $p$)
       because the use of $x$ in $y$'s body is $\omega$. 
\end{itemize}
A difficulty here is that $\pi_y$ needs to be determined at the definition of $y$, while the constraint on $\pi_y$ is only obtained from the use of $y$.

Our design choice is the latter; the multiplicity of a generalizable $\LET$-bound variable is $\omega$ in the system. 
One justification for this choice is that a motivation of polymorphic typing is to enhance reusability, while reuse is not possible 
for variables with multiplicity $1$. 
Another justification is compatibility with recursive definitions, where recursively-defined variables must have multiplicity $\omega$; 
it might be confusing, for example, if the multiplicity of a list-manipulation function changes after 
we change its definition from an explicit recursion to $\var{foldr}$.

\subsection{Inference Rule for Lets}

In summary, the following are our criteria about $\LET$ generalization: 
\begin{itemize}
 \item Only $\LET$s with polymorphic type annotations are generalized.
 \item Variables introduced by $\LET$ to be generalized have multiplicity $\omega$. 
\end{itemize}
This idea can be represented by the following typing rule:
\[
\ninfer{LetA}
{
\ITy{\GG}{\omega \GD_1 + \GD_2}{\LET~x : (\forall \V{p}\V{a}.Q \To \Gt) = e_1~\IN~e_2}{\Gt_2}{C'_1 \wedge C_2}
}
{\bbc
\ITy{\GG}{\GD_1}{e_1}{\Gt_1}{C_1} \quad \set{ \V{\pi}\V{\Ga} } = \fuv(\Gt_1, C_1) \setminus \fuv(\GG) \\
C'_1 = \exists \V{\pi}\V{\Ga}.(Q \models^{\Gt_1} C_1 \wedge \Gt \eqty \Gt_1)
\\
\ITy{\GG\Gth_1,x:(\forall \V{p}\V{a}.Q \To \Gt)}{\GD_2,\Used{x}{M}}{e_2}{\Gt_2}{C_2}
\ee
}
\]
(We do not discuss non-generalizable $\LET$ because they are typed as $(\Gl x.e_2) \A e_1$.)
Constraints like $\exists \V{\pi}\V{\Ga}.(Q \models^{\Gt_1} C_1 \wedge \Gt \eqty \Gt_1)$ above are called \emph{implication constraints}~\cite{VytiniotisJSS11}, 
which states that the entailment must hold only by instantiating unification variables in $\V{\pi}\V{\Ga}$. 
There are two roles of implication constraints. One is to delay the checking because $\Gt_1$ and $C_1$ contain some unification variables
that will be made concrete after this point by solving $C_2$. 
The other is to guard constraints; 
in the above example, since the constraints $C_1 \wedge \Gt \eqty \Gt_1$ hold by assuming $Q$, it is not safe to substitute 
variables outside $\V{\pi}\V{\Ga}$ in solving the constraints because the equivalence might be a consequence of $Q$; recall that $Q$ affects type equality. 
We note that there is a slight deviation from the original approach~\cite{VytiniotisJSS11}; 
an implication constraint in our system is annotated by $\Gt_1$ to identify for which subset of $\{\V{\pi}\V{\Ga}\}$
the existence of a unique solution is not required and thus quantifier elimination is possible, similarly to \cref{sec:defaulting}.

\subsection{Solving Constraints}

Now, the set of constraints is extended to include implication constraints.
\[
 C \DEq \bigwedge_i \psi_i \qquad \psi_i \DEq \dots \mid \exists \V{\pi}\V{\Ga}.(Q \models^{\Gt} C)
\]
As we mentioned above, an implication constraint $\exists \V{\pi}\V{\Ga}.(Q \models^{\Gt} C)$ means that $Q \models C$ must hold by substituting $\V{\pi}$ and $\V{\Ga}$ with appropriate values, where we do not require uniqueness of solutions for unification variables that do not appear in $\Gt$.
That is, $\SimplE{Q}{C}{\top}{\Gth}{\Gt}$ must hold with $\dom(\Gth) \subseteq \set{\V{\Gp}\V{\Ga}}$.

\newcommand{\SolveQ}[6]{#1.\: #2 \vdashI^{#6}_\mathrm{solv} #3 \leadsto #4 ; #5}
\newcommand{\SimplQ}[6]{#1.\: #2 \vdashI^{#6}_\mathrm{simpl} #3 \leadsto #4 ; #5}

Then, following \textsc{OutsideIn(X)}~\cite{VytiniotisJSS11}, we define the \emph{solving judgment} $\SolveQ{\V{\Gp}\V{\Ga}}{Q}{C}{Q'}{\Gth}{\Gt}$, 
which states that we solve $(Q,C)$ as $(Q',\Gth)$ where $\Gth$ only touches variables in $\V{\Gp}\V{\Ga}$, where $\Gt$ is used for disambiguation (\cref{sec:defaulting}).
Let us write $\ms{impl}(C)$ for all the implication constraints in $C$, and $\ms{simpl}(C)$ for the rest. 
Then, we can define the inference rules for the judgment simply by recursive simplification, similarly to the original~\cite{VytiniotisJSS11}.
\begin{gather*}
\infer
{
\SolveQ{\V{\Gp}\V{\Ga}}{Q}{C}{Q_\mathrm{r}}{\Gth}{\Gt}
}
{
\bbc
\SimplQ{\V{\Gp}\V{\Ga}}{Q}{\ms{simpl}(C)}{Q_\mathrm{r}}{\Gth}{\Gt}
\\
\{ 
\SolveQ{\V{\pi_i}\V{\Ga_i}}{Q \wedge Q_i \wedge Q_\mathrm{r}}{C_i}{\top}{\Gth_i}{\Gt_i}
\}_{ (\exists \V{\pi_i}\V{\Ga_i}.(Q_i \models^{\Gt_i} C_i)) \in \ms{impl}(C\Gth) }
\ee
}
\end{gather*} 
Here, $\SimplQ{\V{\Gp}\V{\Ga}}{Q}{C}{Q_\mathrm{r}}{\Gth}{\Gt}$ is a simplification relation defined 
similarly to $\SimplE{Q}{C}{Q_\mathrm{r}}{\Gth}{\Gt}$ except that we are allowed to touch only variables in $\V{\Gp}\V{\Ga}$.
We omit the concrete rules for this version of simplification relation because they are straightforward except
that unification caused by \rUnify and \rEntailEq and quantifier elimination (\cref{sec:defaulting}) are allowed 
only for variables in $\{\V{\Gp}\V{\Ga}\}$.

Accordingly, we also change the typing rules for bindings to use the solving relation instead of the simplification relation.  
{\small\begin{gather*}
\infer
{
  \ITyP{\GG}{f = e; \mi{prog}}
}
{
  \bbc
  \ITy{\GG}{\GD}{e}{\Gt}{C}
  \quad
  \mathhighlight{\SolveQ{\fuv(C,\Gt)}{\top}{C}{Q}{\Gth}{\Gt}}
  \quad
  \set{\V{\pi}\V{\Ga}} = \fuv(Q, \Gt\Gth)
  \\
  \V{p},\V{a}: \text{fresh}
  \quad
  \ITyP{\GG, f : \forall \V{p}\V{a}. (Q \To \Gt\Gth)[\V{\Subst{\Ga}{a},\Subst{\pi}{p}}]}{\mi{prog}}
  \ee
}
\\
\infer
{
  \ITyP{\GG}{f : (\forall\V{p}\V{a}. Q \To \Gt) = e; \mi{prog}} 
}
{
  \bbc
  \ITy{\GG}{\GD}{e}{\Gs}{C}
  \quad
  \mathhighlight{\SolveQ{\fuv(C,\Gs)}{Q}{C \wedge \Gt \eqty \Gs}{\top}{\Gth}{\Gs}}
  \quad
  \ITyP{\GG, f : \forall\V{p}\V{a}. Q \To \Gt}{\mi{prog}}
  \ee
}
\end{gather*}}Above, there are no unification variables other than $\fuv(C,\Gt)$ or $\fuv(C,\Gs)$.

The definition of the solving judgment and the updated inference rules for programs are the same as those in the original \textsc{OutsideIn(X)}~\cite{VytiniotisJSS11} except $\Gt$ for disambiguation. 
This is one of the advantages of being based on \textsc{OutsideIn(X)}.

 \section{Implementation and Evaluation}
\label{sec:evaluation}

In this section, we evaluate the proposed inference method using 
our prototype implementation. 
We first report what types are inferred for functions from \texttt{Prelude} to see 
whether or not inferred types are reasonably simple. 
We then report the performance evaluation that measures efficiency of type inference 
and the overhead due to entailment checking and quantifier elimination.

\subsection{Implementation} 
The implementation follows the present paper except for a few points. 
Following the implementation of \textsc{OutsideIn(X)} in GHC,
our type checker keeps a natural number, which we call an implication level, corresponding to the depth of implication constraints, 
and a unification variable also accordingly keeps the implication level at which the variable is introduced. 
As usual, we represent unification variables by mutable references. 
We perform unification on the fly by destructive assignment, while 
unification of variables that have smaller implication levels than the current level is recorded for later checking of implication constraints;
such a variable cannot be in $\V{\pi}\V{\Ga}$ of $\exists \V{\pi}\V{\Ga}.Q \models^\Gt C$.
The implementation supports GADTs because they can be implemented rather easily by extending constraints $Q$ to include type equalities, 
but does not support type classes because the handling of them requires another \textsc{X} of \textsc{OutsideIn(X)}.

Although we can use a linear-time Horn SAT solving algorithm~\cite{DowlingG84} for checking $Q \models \phi$, 
the implementation uses a general SAT solver based on DPLL~\cite{DavisP60,DavisLL62} because 
the unit propagation in DPLL works efficiently for Horn formulas. 
We do not use external solvers, such as Z3, 
as we conjecture that the sizes of formulas are usually small, and overhead to use external solvers would be high.

\subsection{Functions from Prelude}

We show how our type inference system works for some polymorphic functions from Haskell's \texttt{Prelude}.
Since we have not implemented type classes and I/O in our prototype implementation
and since we can define copying or discarding functions for concrete first-order datatypes, 
we focus on the unqualified polymorphic functions.
Also, we do not consider the functions that are 
obviously unrestricted, such as $\var{head}$ and $\var{scanl}$, in this examination. 
In the implementation of the examined functions, we use natural definitions as possible. For example, 
a linear-time accumulative definition is used for $\var{reverse}$.
Some functions can be defined by 
both explicit recursions and $\var{foldr}$/$\var{foldl}$; among the examined functions, 
$\var{map}$, $\var{filter}$, $\var{concat}$, and $\var{concatMap}$ can be defined by $\var{foldr}$, 
and $\var{reverse}$ can be defined by $\var{foldl}$. 
For such cases, both versions are tested.

\begin{figure}[t]
\[\begin{tarray}{r@{$\,{}:{}\,$}l}
  (\circ)    & (q \le s \wedge q \le t \wedge p \le t) \To (b \pto{q} c) \pto{r} (a \pto{p} b) \pto{s} a \pto{t} c \\
\var{curry}  & (p \le r \wedge p \le s) \To ((a \ox b) \pto{p} c) \pto{q} a \pto{r} b \pto{s} c \\
\var{uncurry}& (p \le s \wedge q \le s) \To (a \pto{p} b \pto{q} c) \pto{r} (a \ox b) \pto{s} c \\
\var{either} & (p \le r \wedge q \le r) \To (a \pto{p} c) \pto{\omega} (b \pto{q} c) \pto{\omega} \con{Either} \A a \A b \pto{r} c \\
\var{foldr}  & (q \le r \wedge p \le s \wedge q \le s) \To (a \pto{p} b \pto{q} b) \pto{\omega} b \pto{r} \con{List} \A a \pto{s} b\\
\var{foldl}  & (p \le r \wedge r \le s \wedge q \le s) \To (b \pto{p} a \pto{q} b) \pto{\omega} b \pto{r} \con{List} \A a \pto{s} b\\
\var{map}    & (p \le q) \To (a \pto{p} b) \pto{\omega} \con{List} \A a \pto{q} \con{List} \A b \\
\var{filter} & (a \pto{p} \con{Bool}) \pto{\omega} \con{List} \A a \pto{\omega} \con{List} \A a \\
\var{append} & \con{List} \A a \pto{p} \con{List} \A a \pto{q} \con{List} \A a \\
\var{reverse} & \con{List} \A a \pto{p} \con{List} \A a\\
\var{concat} & \con{List} \A (\con{List} \A a) \pto{p} \con{List} \A a\\
\var{concatMap} & (p \le q) \To (a \pto{p} \con{List} \A b) \pto{\omega} \con{List} \A a \pto{q} \con{List} \A b
\end{tarray}\]
\caption{Inferred types for selected functions from \texttt{Prelude} (quantifications are omitted)}
\label{fig:inference-results}
\end{figure}

\cref{fig:inference-results} shows the inferred types for the examined functions. 
Since the inferred types coincide for the two variations (by explicit recursions or by folds) of $\var{map}$, $\var{filter}$, $\var{append}$, 
$\var{reverse}$, $\var{concat}$, and $\var{concatMap}$, the results do not refer to these variations. 
Most of the inferred types look unsurprising, considering the fact that the constraint $p \le q$ 
is yielded usually when an input that corresponds to $q$ is used in an argument that corresponds to $p$.
For example, consider $\var{foldr} \A f \A e \A \var{xs}$. 
The constraint $q \le r$ comes from the fact that $e$ (corresponding to $r$) is passed as the second argument of $f$ (corresponding to $q$) via a recursive call.
The constraint $p \le s$ comes from the fact that the head of $\var{xs}$ (corresponding to $s$) is used as the first argument of $f$ (corresponding to $p$). 
The constraint $q \le s$ comes from the fact that the tail of $\var{xs}$ is used in the second argument of $\var{f}$.
A little explanation is needed for 
the constraint $r \le s$ in the type of $\var{foldl}$, 
where both $r$ and $s$ are associated with types with the same polarity. 
Such constraints usually come from recursive definitions. Consider the definition of $\var{foldl}$:
\[
\var{foldl} = \Gl f. \Gl e. \Gl x. \CASE~x~\OF~\{ \con{Nil} \to e; \con{Cons} \A a \A y \to \var{foldl} \A f \A (f \A e \A a) \A y \}
\]
Here, we find that $a$, a component of $x$ (corresponding to $s$), appears in the second argument of $\var{fold}$ (corresponding to $r$), which yields the constraint $r \le s$. 
Note that the inference results do not contain $\pto{1}$; recall that there is no problem in using unrestricted inputs linearly, and 
thus the multiplicity of a linear input can be arbitrary. 
The results also show that the inference algorithm successfully detected that $\var{append}$, $\var{reverse}$, and $\var{concat}$ are linear functions. 

It is true that these inferred types indeed leak some internal details into their constraints, but those constraints can be understood only from their
extensional behaviors, at least for the examined functions. 
Thus, we believe that the inferred types are reasonably simple.

\begin{wraptable}[8]{r}{7.2cm}\newcommand{\pz}{\phantom{0}}
 \small\vspace{-3\baselineskip}\caption{Experimental results}\label{table:experiments}
 \centering\begin{tabular}{l|r|r|r@{~(}r@{)\,}|r@{~(}r@{)}}
                    & & \multicolumn{1}{c|}{Total}           & \multicolumn{2}{c|}{SAT} & \multicolumn{2}{c}{QE} \\[-2pt]
   Program          & LOC & Elapsed           & Elapsed & \# & Elapsed & \# \\\hline
   \texttt{funcs} & 40 & 4.3\pz             & 0.70\pz   & 42                 & 0.086 & 15\\ 
   \texttt{gv}      & 53 & 3.9\pz           & 0.091  & 9                  & 0.14\pz & 17 \\
   \texttt{app1}    & 4 & 0.34              & 0.047  & 4                  & 0.012 & 2\\
   \texttt{app10}   & 4 & 0.84              & 0.049  & 4                  & 0.038 & 21
 \end{tabular}
  (times are measured in ms)
\end{wraptable}

\subsection{Performance Evaluation}

We measured the elapsed time for type checking and the overhead of implication checking and quantifier elimination. 
The following programs were examined in the experiments:
\texttt{funcs}: the functions in \cref{fig:inference-results}, 
\texttt{gv}: an implementation of a simple communication in a session-type system GV~\cite{LindleyM15}
   taken from \cite[Section 4]{LindleyM16} with some modifications,\footnote{
We changed the type of $\var{fork} : \con{Dual} \A s \A s' \pto{\omega} (\con{Ch} \A s \pto{1} \con{Ch} \A \con{End}) \pto{1} (\con{Ch} \A s' \pto{1} \con{Un} \A r) \pto{1} r$, 
as their type $\con{Dual} \A s \A s' \To (\con{Ch} \A s \pto{1} \con{Ch} \A \con{End}) \pto{1} \con{Ch} \A s'$ is incorrect for
the multiplicity erasing semantics. A minor difference is that we used a GADT to witness duality because our prototype implementation does not support type classes. 
}
\texttt{app1}: a pair of the definitions of $\var{app}$ and $\var{app'}$, and 
\texttt{app10}: a pair of the definitions of $\var{app}$ and $\var{app10} = \Gl f. \Gl x. \underbrace{\var{app} \A \dots \A \var{app}}_{10} \A f \A x$.
The former two programs are intended to be miniatures of typical programs.
The latter two programs are intended to measure the overhead of quantifier elimination. 
Although the examined programs are very small, they all involve the ambiguity issues. 
For example, consider the following fragment of the program $\texttt{gv}$:
\[\begin{minipage}{0.95\textwidth}\small\begin{alltt}
answer : Int = fork prf calculator \DOLLAR \textbackslash{}c -> left c & \textbackslash{}c ->
               send (MkUn 3) c & \textbackslash{}c -> send (MkUn 4) c & \textbackslash{}c ->
               recv c & \textbackslash(MkUn z, c) -> wait c & \textbackslash() -> MkUn z
\end{alltt}\end{minipage}\]
(Here, we used our paper's syntax instead of that of the actual examined code.)
Here, both \texttt{\DOLLAR} and \texttt{\&} are operator versions of $\var{app}$, where the arguments are flipped in \texttt{\&}. 
As well as treatment of multiplicities, the disambiguation is crucial for this expression to have type $\con{Int}$.

The experiments were conducted on a MacBook Pro (13-inch, 2017) with Mac OS 10.14.6,  
3.5 GHz Intel Core i7 CPU, and 16 GB memory. 
GHC 8.6.5 with \texttt{-O2} was used for compiling our prototype system.

\Cref{table:experiments} lists the experimental results. Each elapsed time is the average of 1,000 executions for the first 
two programs, and 10,000 executions for the last two. All columns are self-explanatory except for the \# column, which counts 
the number of executions of corresponding procedures. 
We note that the current implementation restricts $Q_\mr{w}$ in {\rEntail} to be $\top$ and removes redundant constraints 
afterward. 
This is why the number of SAT solving in \texttt{app1} is four instead of two. 
For the artificial programs (\texttt{app1} and \texttt{app10}), the overhead is not significant; typing cost grows faster than SAT/QE costs. 
In contrast, the results for the latter two show that SAT becomes heavy for
higher-order programs (\texttt{funcs}), and quantifier elimination becomes heavy for combinator-heavy programs (\texttt{gv}), 
although we believe that the overhead would still be acceptable. 
We believe that, since we are currently using naive algorithms for both procedures, there is much room to reduce the overhead. 
For example, if users annotate most general types, the simplification invokes trivial checks $\bigwedge_i \phi_i \models \phi_i$ often. 
Special treatment for such cases would reduce the overhead.

 \section{Related Work}
\label{sec:related}

Borrowing the terminology from Bernardy \etal~\cite{BeBNJS18},
there are two approaches to linear typing: linearity via arrows and linearity via kinds.
The former approaches manage how many times an assumption (\ie, a variable) can be used; for example, in Wadler~\cite{Wadler93}'s 
linear $\Gl$ calculus, there are two sort of variables: linear and unrestricted, where the latter variables can only be obtained by 
decomposing $\LET~!x = e_1~\IN~e_2$. 
Since primitive sources of assumptions are arrow types, it is natural to annotate them with arguments' multiplicities~\cite{BeBNJS18,McBride16,GhicaS14}.
For multiplicities, we focused on $1$ and $\omega$ following Linear Haskell~\cite{BeBNJS18,LinearMiniCore,LinearHaskellProposal}.
Although $\set{1,\omega}$ would already be useful for some domains including 
reversible computation~\cite{Janus,YoAG11} and quantum~computation\cite{AltenkirchG05,SelingerV06},
handling more general multiplicities, such as $\set{0,1,\omega}$ and 
arbitrary semirings~\cite{GhicaS14}, is an interesting future direction. 
Our discussions in \cref{sec:lang,sec:inference}, similarly to Linear Haskell~\cite{BeBNJS18}, could be extended to more general domains
with small modifications. 
In contrast, we rely on the particular domains $\set{1,\omega}$ of multiplicities
for the crucial points of our inference, \ie, entailment checking and quantifier elimination. 
Igarashi and Kobayashi~\cite{IgarashiK00}'s linearity analysis for $\pi$ calculus, 
which assigns input/output usage (multiplicities) to channels, has similarity to linearity via arrows.
Multiplicity $0$ is important in their analysis to identify input/output only channels. 
They solve constraints on multiplicities separately in polynomial time, leveraging monotonicity of multiplicity operators with respect to ordering $0 \le 1 \le \omega$. 
Here, $0 \le 1$ comes from the fact that $1$ in their system means ``at-most once'' instead of ``exactly once''.

The ``linearity via kinds'' approaches distinguish types of which values are treated linearly and types of which values are not~\cite{MaZZ10,TovP11,Morris16}, 
where the distinction usually is represented by kinds~\cite{MaZZ10,TovP11}.
Interestingly, they also have two function types---function types that belong to the linear kind and those that belong to the unrestricted kind---because 
the kind of a function type cannot be determined solely by the argument and return types. 
Mazurak~\etal~\cite{MaZZ10} use subkinding to avoid explicit conversions from unrestricted values to linear ones.
However, due to the variations of the function types, a function can have multiple incompatible types; \eg, the function $\var{const}$ can have four incompatible types~\cite{Morris16} in the system. 
Universal types accompanied by kind abstraction~\cite{TovP11} address the issue to some extent; it works well for $\var{const}$,
but still gives two incomparable types to the function composition $(\circ)$~\cite{Morris16}.
Morris~\cite{Morris16} addresses this issue of principality with qualified typing~\cite{QualifiedTypes}.
Two forms of predicates are considered in the system: $\con{Un} \A \Gt$ states that $\Gt$ belongs to the unrestricted kind,
and $\Gs \le \Gt$ states that $\con{Un} \A \Gs$ implies $\con{Un} \A \Gt$. 
This system is considerably simple compared with the previous systems. Turner~\etal~\cite{TurnerWM95}'s type-based usage analysis has a similarity to the linearity via kinds; in the system, each type is annotated by usage (a multiplicity) as $(\con{List} \A \con{Int}^{\omega})^{\omega}$.
Wansbrough and Peyton Jones~\cite{WansbroughJ99} extends the system to include polymorphic types and subtyping with respect to multiplicities, 
and have discussions on multiplicity polymorphism. 
Mogensen~\cite{Mogensen97} is a similar line of work, which reduces constraint solving on multiplicities to Horn SAT. 
His system concerns multiplicities $\{0,1,\omega\}$ with ordering $0 \le 1 \le \omega$, 
and his constraints can involve more operations including additions and multiplications but only in the left-hand side of $\le$.

Morris~\cite{Morris16} uses improving substitutions~\cite{Jones95} in generalization, which sometimes are
effective for removing ambiguity, though without showing concrete algorithms to find them.
In our system, as well as {\rEntailEq}, 
$\Elim{\pi}{Q}$ can be viewed as a systematic way to find improving substitutions. 
That is, $\Elim{\pi}{Q}$ improves $Q$ by 
substituting $\pi$ with $\min \{ M_i \mid \omega \le M_i \in \Phi_\omega \}$, \ie, the largest possible candidate of $\pi$. 
Though the largest solution is usually undesirable, especially when the right-hand sides of $\le$ are all singletons, 
we can also view that $\Elim{\pi}{Q}$ substitutes $\pi$ by $\prod_{ \mu_i \le 1 \in \Phi_1 } \mu_i$, \ie, the smallest possible candidate.

 \section{Conclusion}
\label{sec:conclusion}

We designed a type inference system for a rank 1 fragment of $\Gl^q_\to$~\cite{BeBNJS18} that can infer principal types
based on the qualified typing system \textsc{OutsideIn(X)}~\cite{VytiniotisJSS11}. 
We observed that naive qualified typing infers ambiguous types often 
and addressed the issue based on quantifier elimination. 
The experiments suggested that the proposed inference system infers principal types effectively, and the overhead 
compared with unrestricted typing is acceptable, though not negligible. 

Since we based our work on the inference algorithm used in GHC, 
the natural expectation is to implement the system into GHC. 
A technical challenge to achieve this is combining the disambiguation techniques with 
other sorts of constraints, especially type classes, and arbitrarily ranked polymorphism.

\subsection*{Acknowledgments}

We thank Meng Wang, Atsushi Igarashi, and the anonymous reviewers of ESOP 2020 for their helpful comments on the preliminary versions of this paper. 
This work was partially supported by JSPS KAKENHI Grant Numbers 15H02681 and 19K11892, 
JSPS Bilateral Program, Grant Number JPJSBP120199913, 
the Kayamori Foundation of Informational Science Advancement, 
and EPSRC Grant \emph{EXHIBIT: Expressive High-Level Languages for Bidirectional Transformations} (EP/T008911/1).

\bibliographystyle{splncs04}

\iffullversion
\appendix

\clearpage
\section{Limitation of Multiplicity Handling in Linear Haskell}
\label{sec:limitation-LinearHaskell}

For very simple cases, there is no problem with type checking of multiplicity-polymorphic functions in Linear Haskell.
\begin{alltt}
Prelude> :\{
Prelude| id1 :: a -->.(p) a
Prelude| id1 x = x
Prelude| :\} 
Prelude> -- No type errors 
\end{alltt}
Here, \verb|a -->.(p) b| corresponds to $a \pto{p} b$ in our notation.
If we ask its type, we will see an unrestricted one due to defaulting. 
\begin{alltt}
Prelude> :type id1
id1 :: a -> a
\end{alltt}
Notice that \verb|a -> b| corresponds to our $a \pto{\omega} b$.  
However, the following example ensures that the \verb|id1| is indeed polymorphic about its argument's multiplicity:
\begin{alltt}
Prelude> :\{ 
Prelude| idL :: a ->. a
Prelude| idL = id1 
Prelude| :\}
Prelude> -- No type errors
\end{alltt}
Here, \verb|a ->. b| corresponds to $a \pto{1} b$.

Now, let us consider the function composition. 
If we define it in Linear Haskell, we obtain the version that can work only 
for unrestricted functions by defaulting. 
\begin{alltt}
Prelude> let comp f g x = f (g x)
Prelude> :type comp
comp :: (t1 -> t2) -> (t3 -> t1) -> t3 -> t2
\end{alltt}
As its type suggests, we are not able to compose linear functions by \verb|comp|. 
It is true that by giving its type explicitly we can define the linear function composition. 
\begin{alltt}
Prelude> :\{
Prelude| comp :: (b ->. c) -> (a ->. b) -> (a ->. c)
Prelude| comp f g x = f (g x)
Prelude| :\}
\end{alltt}
Then, we are tempted to parameterize the multiplicity to obtain a slightly general (but not most general) definition, 
but this attempt fails. 
\begin{alltt}
Prelude> :\{
Prelude| comp :: (b -->.(r) c) -> (a -->.(r) b) -> (a -->.(r) c)
Prelude| comp f g x = f (g x)
Prelude| :\}

<interactive>:30:10: error:
    • Nontrivial multiplicity equalities are currently not supported
    ...

<interactive>:30:10: error:
    • Couldn't match type ‘r’ with ‘'Omega’
    ...
\end{alltt}
Notice that the principal type for $\var{comp}$ in $\Gl^q_\to$ is $(b \pto{p_1} c) \pto{p_2} (a \pto{p_3} b) \pto{p_1 \cdot q} a \pto{p_1 \cdot p_3 \cdot r} c$.\footnote{Technically, we need a small correction to the typing rule: the typing rule for variables must be changed as \cite{LinearMiniCore} for $\var{comp}$ to have the type.}
However, as far as we have investigated, it seems that Linear Haskell does not provide any ways for users to write multiplication of multiplicities. 

Actually, in their implementation note, it is explicitly said that they compromise on solving constraints\footnote{\url{https://gitlab.haskell.org/ghc/ghc/wikis/linear-types/implementation}}.

\section{Selected Proofs}

\subsection{Proof of \cref{thm:correspondence-with-original}}
\label{proof:correspondence-with-original}

We first give a rank 1 fragment of the original type system of $\Gl^q_\to$~\cite{BeBNJS18} (with the modification in \cite{LinearMiniCore} for the variable case) 
in \cref{fig:original-typing}. 
In addition to the treatment of variable expressions, this version has two notable presentational differences from the original.
First, to make our discussions easy, we use two environments instead of one in \cref{fig:original-typing}, 
as well as in \cref{fig:typing}.
Second, since we consider rank 1 fragment, 
the multiplicity application rule (and the type application rule, which is omitted in the original presentation~\cite{BeBNJS18}) is 
merged into {\rVar}. 
As we have mentioned in \cref{sec:meta}, we use the condition $\Used{x}{1} \le \GD$ 
instead of $\exists \GD'.\: \GD = \Used{x}{1} + \omega \GD'$
in the premise of {\rVar}. 

One would find that the rules for $\NTy{\GG}{\GD}{e}{\Gt}{Q}$ with $Q = \top$ (shown in \cref{fig:typing})
and the rules for $\OTy{\GG}{\GD}{e}{\Gt}$ (shown in \cref{fig:original-typing}) look 
almost the same. 
One difference is that, 
since the original system allows that multiplications of multiplications directly appear in types, some $\mu$s in \cref{fig:typing}
are replaced with $M$s in \cref{fig:original-typing}. This difference is rather subtle, 
as we do not use this extra flexibility in the proof. 
Another difference is in {\rCase}: 
while in the former system each branch $e_i$ can be typechecked under a different multiplicity environment $\GD_i$, 
in the latter system each branch is checked under the same multiplicity environment.
To absorb the difference, together with the fact $\models \GD_i \le \bigsqcup_j \GD_j$, 
we use a stronger variant of \cref{lemma:weakening-multiplicity} 
that states in addition that the derivations for $\NTy{\GG}{\GD}{e}{\Gt}{Q}$ and $\NTy{\GG}{\GD'}{e}{\Gt}{Q}$
have the same height. 

Now, we are ready to state \cref{thm:correspondence-with-original} more formally.
\begin{theorem*}
Suppose that $\GG(x)$ is monotype for all $x$. Then, 
$\NTy{\GG}{\GD}{e}{\Gt}{\top}$ implies $\OTy{\GG}{\GD}{e}{\Gt}$.
\end{theorem*}

\begin{proof}
Straightforward induction on the height of the derivation of $\NTy{\GG}{\GD}{e}{\Gt}{\top}$, 
appealing to the stronger version of \cref{lemma:weakening-multiplicity} above for {\rCase}. \qed 
\end{proof}

\begin{figure}[t]\small
\setlength{\jot}{1.4ex}
\setlength{\abovedisplayskip}{0pt}
\setlength{\belowdisplayskip}{0pt}
\begin{gather*}
\ninfer{\rVar}
{
\OTy{\GG}{\GD}{x}{\Gt[\V{\Subst{p}{M}}, \V{\Subst{a}{\Gt}}]}
}
{
\bbc
\GG(x) = \forall \V{p} \V{a}. \Gt \quad
\Used{x}{1} \le \GD 
\ee
}
\\
\ninfer{\rAbs}
{
\OTy{\GG}{\GD}{\Gl x. e}{\Gs \pto{M} \Gt}
}
{
\OTy{\GG,x:\Gs}{\GD, \Used{x}{M}}{e}{\Gt}
}
\qquad
\ninfer{\rApp}
{
\OTy{\GG}{\GD_1 + M\GD_2}{e_1 \A e_2}{\Gt}
}
{
\OTy{\GG}{\GD_1}{e_1}{\Gs \pto{M} \Gt}
&
\OTy{\GG}{\GD_2}{e_2}{\Gs}
}
\\
\ninfer{\rCon}
{
\OTy{\GG}{\omega\GD_0 + \sum_i \nu_i[\V{\Subst{p}{M}}] \GD_i}{\con{C} \A \V{e}}{\con{D} \A \V{M} \A \V{\Gs}}
}
{
 \con{C} : \forall \V{p} \V{a}.\: \V{\Gt} \pto{\V{\nu}} \con{D} \A \V{p} \A \V{a} &
\{ \OTy{\GG}{\GD_i}{e_i}{\Gt_i[ \V{\Subst{p}{M}}, \V{\Subst{a}{\Gs}}] } \}_i
}
\\
\ninfer{\rCase}
{
  \OTy{\GG}{\mu_0 \GD_0 + \GD}{\CASE~e_0~\OF~\{ \con{C}_i \A \V{x_i} \to e_i \}_i}{\Gt'}
}
{
  \bb
  \OTy{\GG}{\GD_0}{e_0}{\con{D} \A \V{M} \A \V{\Gs}} \\\left\{
  \bb
   \con{C}_i : \forall \V{p}\V{a}. \: \V{\Gt_i} \pto{\V{\nu_i}} \con{D} \A \V{p} \A \V{a} \\
  \OTy{\GG,\V{x_i:\Gt_i[\V{\Subst{p}{M}},\V{\Subst{a}{\Gs}}]}}{\GD, \V{ \Used{x_i}{\mu_0 \nu_i[\V{\Subst{p}{M}}]} }}{e_i}{\Gt'}
\ee
  \right\}_i
\ee
}
\end{gather*}
\caption{The original typing rules adapted for the expressions in \cref{sec:lang}}
\label{fig:original-typing}
\end{figure}

\subsection{Proof of \cref{lemma:simplification-sound}.}
\label{proof:simplification-sound}

The proof is done by the induction on the derivation. 
Let us write $\mc{E}_\Gth$ for $\bigwedge_{\pi \in \dom(\Gth)} (\pi = \Gth(\pi)) \wedge \bigwedge_{\Ga \in \dom(\Gth)} (\Ga \eqty \Gth(\Ga))$.
Then, our goal is to prove that, if $\Simpl{Q}{C}{Q'}{\Gth}$, then 
$Q \wedge Q' \models C\Gth$, $\Disjoint{\dom(\Gth)}{\fuv(Q)}$, 
$\Disjoint{\dom(\Gth)}{\fuv(Q')}$, $Q \wedge C \models Q' \wedge \mc{E}_\Gth$ hold.
Let us write $C \equiv C'$ if both $C \models C'$ and $C' \models C$ hold. 

\Case{\rname{S-Fun}}
This case is straightforward because we have 
$(\Gs \pto{\mu} \Gt) \eqty (\Gs' \pto{\mu'} \Gt') \equiv {\Gs \eqty \Gs' \wedge \mu \le \mu' \wedge \mu' \le \mu \wedge \Gt \eqty \Gt'}$.

\Case{\rname{S-Data}}
Similarly, this case is also straightforward because we have
$(\con{D} \A \V{\mu} \A \V{\Gs}) \eqty (\con{D} \A \V{\mu'} \A \V{\Gs'}) \equiv \V{\mu \le \mu'} \wedge \V{\mu' \le \mu} \wedge \V{\Gs \eqty \Gs'}$.

\Case{\rname{S-Unify}}
In this case, we have $C = \Ga \eqty \Gt \wedge C$, $\Gth = \Gth' \circ [\Subst{\Ga}{\Gt}$, $\Ga \not\in \fuv(\Gt)$, and
$\Simpl{Q}{C[\Subst{\Ga}{\Gt}]}{Q'}{\Gth'}$. 
By the induction hypothesis, we have that $(Q',\Gth')$ is a guess-free solution for $(Q,C[\Subst{\Ga}{\Gt}])$.
Then, we have $Q \models C[\Subst{\Ga}{\Gt}]\Gth'$, and thus $Q \models C\Gth$ by definition of $\Gth$. 
Notice that we have $\dom(\Gth) = \dom(\Gth') \cup \{\Ga\}$.
Since $\Ga$ is a unification type variable, it cannot appear in both $Q$ and $Q'$.
Thus, we have $\Disjoint{\dom(\Gth)}{\fuv(Q)}$ and
$\Disjoint{\dom(\Gth)}{\fuv(Q')}$.
Then, we will show the guess-freeness. From the induction hypothesis, we have $Q \wedge C[\Subst{\Ga}{\Gt}] \models Q' \wedge \mc{E}_{\Gth'}$.
Notice that we have $\mc{E}_{\Gth} \equiv \mc{E}_{\Gth'} \wedge \Ga \eqty \Gt$, and
$C \wedge \Ga \eqty \Gt \models C[\Subst{\Ga}{\Gt}]$.
Thus, we have $Q \wedge \Ga \eqty \Gt \wedge C \models Q' \wedge \mc{E}_{\Gth}$. 

\Case{\rname{S-Triv}}
Similarly, this case is also straightforward because we have
${\Gt \eqty \Gt} \equiv \top$.

\Case{\rEntail}
In this case, we have $C = \phi \wedge Q_\mr{w} \wedge C$, $Q \wedge Q_\mr{w} \models \phi$ and 
$\Simpl{Q}{Q_\mr{w} \wedge C}{Q'}{\Gth}$.
We first show that $(Q',\Gth)$ is a solution for $(Q,C)$.
By the induction hypothesis, we have 
$Q \wedge Q' \models (Q_\mr{w} \wedge C)\Gth$, 
$\Disjoint{\dom(\Gth)}{\fuv(Q)}$, and 
$\Disjoint{\dom(\Gth)}{\fuv(Q')}$. 
By $Q \wedge Q_\mr{w} \models \phi$, we have $(Q \wedge Q_\mr{w}) \Gth\models \phi\Gth$, and thus $Q \wedge Q_\mr{w}\Gth \models \phi\Gth$. 
Thus, we have $Q \wedge Q' \models Q_\mr{w} \Gth \wedge \phi \Gth \wedge \Gth$. 
Then, we prove the guess-freeness of the solution. 
By the induction hypothesis, we have 
$Q \wedge Q_\mr{w} \wedge C \models Q' \wedge \mc{E}_{\Gth}$. 
Then, it trivially holds that $Q \wedge \phi \wedge Q_\mr{w} \wedge C \models Q' \wedge \mc{E}_{\Gth}$.

\Case{\rRem}
Trivial.
\qed 

\subsection{Proof of \cref{lemma:simplification-complete}.}
\label{proof:simplification-complete}

We prove the contraposition of the statement. 
Suppose that we do not have $(Q'',\Gth')$ such that 
$\Simpl{Q}{C}{Q''}{\Gth'}$. Then, it must be the case that $C$ contains $\Gt \eqty \Gt'$ (modulo commutativity of $\eqty$) such that 
either of the following holds. 
\begin{itemize}
 \item $\Gt = \Gt_1 \pto{\mu_1} \Gt_2$ and $\Gt' = \con{D}' \A \V{\mu'} \A \V{\Gs'}$, 
 \item $\Gt = \con{D} \A \V{\mu} \A \V{\Gs}$ and $\Gt' = \con{D}' \A \V{\mu'} \A \V{\Gs'}$ and $\con{D} \ne \con{D}'$, and 
 \item $\Gt = \Ga \ne \Gt'$ and $\Ga \in \fuv(\Gt')$. 
\end{itemize}
In either case, there is no solution for $(Q,C)$ when $Q$ is satisfiable. 
\qed

\subsection{Proof of \cref{theorem:inference-sound}.}
\label{proof:inference-sound}

We prove the statement by using the induction on the structure of $e$. 
Let us write $\Gz$ for unification variables that is either $\Gp$ or $\Ga$.

\Case{$e = x$} 
In this case, we have $\GG(x) = \forall \V{p}\V{a}. Q_x \To \Gt_x$, $\Gt = \Gt_x[\V{\Subst{p}{\pi}},\V{\Subst{a}{\Ga}}]$, 
$\GD = \set{\Used{x}{1}}$, and $C = Q_x[\V{\Subst{p}{\pi}}]$.
By the assumption that $(Q,\Gth)$ is a solution for $C$ under $\top$, 
we have $Q \models Q_x[\V{\Subst{p}{\pi}}]\Gth$.
Here, we have $(\GG\Gth)(x) = \forall \V{p}\V{a}. Q_x\Gth \To \Gt_x\Gth$, 
where we can assume that $\V{p},\V{a} \not\in \ftv(\Gth(\Gz))$ for any $\Gz \in \dom(\Gth)$, and thus 
$\Gt\Gth$ can be written as $\Gt_x\Gth[\VSubst{p}{\pi\Gth},\VSubst{a}{\Ga\Gth}]$. Thus, by using \rVar, we have $\NTy{\GG\Gth}{\GD\Gth}{e}{\Gt\Gth}{Q}$.

\Case{$e = \Gl x.e_1$} 
In this case, we have $\ITy{\GG, x:\Ga}{\GD, \Used{x}{M}}{e_1}{\Gt_1}{C_1}$, $\Gt = \Ga \pto{\pi} \Gt_1$
and $C = C_1 \wedge M \le \pi$.
Suppose that there is a solution $(Q,\Gth)$ for $(\top, C)$. By definition, $(Q,\Gth)$ is 
also a solution for $(\top,C_1)$.
Thus, by the induction hypothesis, we have $\NTy{(\GG, x:\Ga)\Gth}{\GD\Gth, \Used{x}{M\Gth}}{e_1}{\Gt_1\Gth}{Q}$. 
Since we have $Q \models M\Gth \le \pi\Gth$ by $(Q,\Gth)$ is a solution for $(\top,C)$, 
by \cref{lemma:weakening-multiplicity}, we have $\NTy{(\GG, x:\Ga)\Gth}{\GD\Gth, \Used{x}{\pi\Gth}}{e_1}{\Gt_1\Gth}{Q}$. 
Thus, by using \rAbs, we have $\NTy{\GG\Gth}{\GD\Gth}{\Gl x.e_1}{(\Ga \pto{\pi} \Gt_1)\Gth}{Q}$.

\Case{$e = e_1 \A e_2$} In this case, we have
$\ITy{\GG}{\GD_1}{e_1}{\Gt_1}{C_1}$, $\ITy{\GG}{\GD_2}{e_2}{\Gt_2}{C_2}$, 
$\GD = \GD_1 + \pi \GD_2$, 
$\Gt = \Gb$, and 
$C = C_1 \wedge C_2 \wedge \Gt_1 \eqty (\Gt_2 \pto{\pi} \Gb)$.
Let $(Q,\Gth)$ be a solution for $(\top,C)$.
By definition $(Q,\Gth)$ is a solution for $(\top,C_i)$ for each $i = 1,2$. 
Thus, by the induction hypothesis, we have $\NTy{\GG\Gth}{\GD_i\Gth}{e_i}{\Gt_i}{Q}$ for each $i$.
Since we have $Q \models \Gt_1\Gth \eqty (\Gt_2 \pto{\pi} \Gb)\Gth$ by $(Q,\Gth)$ is a solution for $(\top,C)$, 
by using {\rApp} and \rname{Eq}, we have $\NTy{\GG\Gth}{\GD\Gth}{e}{\Gt\Gth}{Q}$.

\Case{$e = \con{C} \A \V{e}$}
We omit the discussion for the case because it is similar to the above cases.

\Case{$e = \CASE~e_0~\{ \con{C}_i \A \V{x_i} \to e_i \}_i$}
In this case, we have 
$\ITy{\GG}{\GD_0}{e_0}{\Gt_0}{C_0}$, and for each $i$, 
\(\con{C}_i : \forall \V{p}\V{a}. \: \V{\Gt_i} \pto{\V{\nu_i}} \con{D} \A \V{p} \A \V{a}\), 
$\ITy{\GG,\V{x_i:\Gt_i[\VSubst{p}{\pi_i},\VSubst{a}{\Ga_i}]}}{\GD_i, \V{\Used{x_i}{M_i}}}{e_i}{\Gt'_i}{C_i}$, 
$\Gt = \Gb$, 
$\GD = \pi_0 \GD_0 + \bigsqcup_i \GD_i$,
\( 
C = C_0 \wedge \bigwedge_i \bigl( C_i \wedge \Gb \eqty \Gt_i \wedge (\Gt_0 \eqty \con{D} \A \V{\pi_i} \A \V{\Ga_i}) \wedge \bigwedge_j M_{ij} \le \pi_0 \nu_{ij}[\VSubst{p}{\pi_i}] \bigr)
\).
Let $(Q,\Gth)$ be a solution for $(\top,C)$, which then is also a solution for $(\top,C_i)$.
Thus, by the induction hypothesis, we have 
$\NTy{\GG\Gth}{\GD_0\Gth}{e_0}{\Gt_0\Gth}{Q}$ and 
$\NTy{\GG\Gth,\V{x_i:\Gt_i[\VSubst{p}{\pi_i},\VSubst{a}{\Ga_i}]\Gth}}{\GD_i\Gth, \V{\Used{x_i}{M_i\Gth}}}{e_i}{\Gt'_i\Gth}{Q}$.
Since we have $Q \models (M_{ij} \le \pi_0 \nu_{ij}[\VSubst{p}{\pi_i}])\Gth$, by \cref{lemma:weakening-multiplicity}, 
we have $\NTy{\GG\Gth,\V{x_i:\Gt_i\Gth[\VSubst{p}{\pi_i\Gth},\VSubst{a}{\Ga_i\Gth}]}}{\GD_i\Gth, \V{\Used{x_i}{\nu_{ij}[\VSubst{p}{\pi_i\Gth}]}}}{e_i}{\Gt'_i\Gth}{Q}$;
notice here that $\V{\Gt_i\Gth[\VSubst{p}{\pi_i\Gth},\VSubst{a}{\Ga_i\Gth}] = \Gt_i[\VSubst{p}{\pi_i},\VSubst{a}{\Ga_i}]\Gth}$
because of the freshness of $\V{p}$ and $\V{a}$.
Thus, by using {\rCase} and \rname{Eq}, 
we have $\NTy{\GG\Gth}{\GD\Gth}{e}{\Gt\Gth}{Q}$.
\qed

\subsection{Proof of \cref{thm:inference-complete}.}
\label{proof:inference-complete}
We prove the statement by using induction on the structure of $e$ (which corresponds to the derivation of  
$\ITy{\GG}{\GD}{e}{\Gt}{C}$ because the rules are syntax-directed). 
We will use the property that $\fuv(C,\GD) \subseteq X \cup \fuv(\GG)$.

\Case{$e = x$}
In this case, we have $\GG(x) = \forall \V{p}\V{a}. Q_x \To \Gt_x$, $\Gt = \Gt_x[\V{\Subst{p}{\pi}},\V{\Subst{a}{\Ga}}]$, 
$\GD = \set{\Used{x}{1}}$, and $C = Q_x[\V{\Subst{p}{\pi}}]$.
By assumption, we have $Q \models \Gt' \eqty \Gt_x\Gth'[\V{\Subst{p}{\mu}},\V{\Subst{a}{\Gs}}]$, $Q' \models Q_x\Gth'[\V{\Subst{p}{\mu}}]$
and $Q' \models \set{ \Used{x}{1}} \le \GD'$.
Take $\Gth = \Gth' \uplus \set{\V{\Subst{\Gp}{\mu}},\V{\Subst{\Ga}{\Gs}}}$. 
Then, $\Gt \Gth = \Gt_x\Gth'[\V{\Subst{p}{\mu}},\V{\Subst{a}{\Gs}}]$ and thus 
$Q' \models \Gt \Gth \eqty \Gt'$. 
Then, we have $C \Gth = Q_x\Gth'[\V{\Subst{p}{\mu}}]$, which implies $(Q', \Gth)$ is a solution for $(\top,C)$. 

\Case{$e = \Gl x.e_1$}
In this case, we have $\Gt = \Ga \pto{\Gp} \Gt_1$, 
$\ITy{\GG,x:\Ga}{\GD, \Used{x}{M} }{e_1}{\Gt_1}{C_1}$ and
$C = C_1 \wedge M \le \Gp$. 
Let $X_1$ be the unification variables introduced in the subderivation. 
Notice that $\Gp, \Ga \not\in X_1$ and $X = \set{\Gp,\Ga} \cup X_1$.

By assumption, we have $\NTy{\GG\Gth'}{\GD'}{\Gl x. e_1}{\Gt'}{Q'}$, 
which implies $Q' \models \Gt' \eqty (\Gs' \pto{\mu'} \Gt_1')$, 
$\NTy{\GG\Gth', x:\Gs'}{\GD'_1}{e_1}{\Gt_1'}{Q'}$, 
and $Q' \models \GD'_1 = (\GD', \Used{x}{\mu'})$. 
Let $\Gth'_1$ be $\Gth' \uplus \set{ \Ga \mapsto \Gs' }$. Notice that $\dom(\Gth'_1) \subseteq \fuv(\GG,x:\Ga)$ holds.

By the induction hypothesis, there exists $\Gth_1$ such that 
$\dom(\Gth_1) \setminus \dom(\Gth_1') \subseteq X_1$, 
$(Q', \Gth_1)$ is a solution for 
$(\top,C_1)$, $Q' \models \Gth_1|_{\dom(\Gth_1')} = \Gth_1'$, $Q' \models \Gt_1 \Gth_1 \eqty \Gt_1'$
and $Q' \models (\GD, \Used{x}{M})\Gth_1 \le (\GD', \Used{x}{\mu'})$. 
Let $\Gth$ be $\Gth_1 \uplus \{ \pi \mapsto \mu' \}$. 
Here, $\dom(\Gth) \setminus \dom(\Gth') = (\dom(\Gth_1) \cup \set{\pi}) \setminus (\dom(\Gth'_1) \setminus \set{\Ga}) \subseteq (\dom(\Gth_1) \setminus \dom(\Gth'_1)) \cup \set{\Gp, \Ga} \subseteq X_1 \cup \set{\Gp,\Ga} = X$ holds. 
Notice that $Q' \models \Gth(\Ga) = \Gth'(\Ga)$ 
as $\Ga \in \dom(\Gth'_1)$. Thus, we have $Q' \models \Gth|_{\dom(\Gth')} = \Gth'$. 
Since $(Q',\Gth_1)$ is a solution for $C_1$, we have $Q' \models C_1 \Gth_1$ and thus $Q' \models C_1 \Gth$.
By $Q' \models (\GD, \Used{x}{M})\Gth_1 \le (\GD', \Used{x}{\mu'})$, 
we have $Q' \models \GD\Gth_1 \le \GD'$ and $Q' \models M\Gth_1 \le \mu$. Since $\GD$ and $M$ cannot contain $\pi$, 
we have $Q' \models \GD\Gth \le \GD'$ and $Q' \models (M \le \pi) \Gth$. 
Hence, by $Q' \models C_1 \Gth$ and $Q' \models (M \le \pi) \Gth$, we have $Q' \models C\Gth$, which implies that $(Q',\Gth)$ is 
a solution for $(\top,C)$.

\Case{$e = e_1 \A e_2$}
In this case, we have 
$\ITy{\GG}{\GD_1}{e_1}{\Gt_1}{C_1}$, 
$\ITy{\GG}{\GD_2}{e_2}{\Gt_2}{C_2}$, 
$\Gt = \Gb$, 
$\GD = \GD_1 + \pi \GD_2$, and 
$C = C_1 \wedge C_2 \wedge \Gt_1 \eqty (\Gt_2 \pto{\pi} \Gb)$
where $\pi$ and $\Gb$ are fresh unification variables. 
Let $X_1$ and $X_2$ be the sets of unification variables introduced in the subderivations for $e_1$ and $e_2$, respectively.
Notice that $X = X_1 \uplus X_2 \uplus \set{\pi, \Gb}$.

By assumption, we have $\NTy{\GG\Gth'}{\GD'}{e_1 \A e_2}{\Gt'}{Q'}$. 
This means that we have 
$\NTy{\GG\Gth'}{\GD'_1}{e_1}{\Gt'_1}{Q'}$, 
$\NTy{\GG\Gth'}{\GD'_2}{e_2}{\Gt'_2}{Q'}$, 
$Q' \models \Gt'_1 \eqty (\Gt'_2 \pto{\mu'} \Gt')$, and 
$Q' \models \GD' = \GD'_1 + \mu'\GD'_2$.

By the induction hypothesis, for each $i = 1,2$, 
we have $\Gth_i$ such that 
$\dom(\Gth_i) \setminus \dom(\Gth') \subseteq X_i$, 
$(Q',\Gth_i)$ is a solution for $(\top,C_i)$, 
$Q' \models \Gth_i|_{\dom(\Gth')} = \Gth'$, 
$Q' \models \Gt_i \Gth_i \eqty \Gt'_i$, and 
$Q' \models \GD_i \Gth_i \le \GD'_i$.
Let us define $\Gth$ as $\Gth_1|_{X_1} \uplus \Gth_2|_{X_2} \uplus \Gth' \uplus \set{ \Subst{\pi}{\mu}, \Subst{\Gb}{\Gt'} }$.
Clearly, we have $\Gth|_{\dom(\Gth')} = \Gth'$ and $\Gt \Gth = \Gb \Gth = \Gt'$. 
By construction, we have $\fuv(C_i, \GD_i) \subseteq X_i \cup \fuv(\GG)$.
Thus, we have $C_i \Gth = C_i \Gth_i$, and thus $\GD_i \Gth = \GD_i \Gth_i$.
By $Q' \models \GD_i \Gth_i \le \GD'_i$, 
we have $Q' \models (\GD_1 + \pi \GD_2) \le \GD'_1 + \mu'\GD'_2$.
By $Q' \models \Gt_i \Gth_i \eqty \Gt'_i$ and $\Gt'_1 \eqty (\Gt'_2 \pto{\mu'} \Gt')$, 
we have $Q' \models (\Gt_1 \eqty (\Gt_2 \pto{\pi} \Gb))\Gth$.
Thus, we have $Q' \models C\Gth$, meaning that $Q'$ is a solution for $(\top,C)$.

\Case{$e = \con{C} \A \V{e}$}
We omit the discussion for the case because it is similar to the above cases.

\Case{$e = \CASE~e_0~\OF~\{ \con{C} \A \V{x_i} \to e_i \}_i$}
In this case, we have 
$\ITy{\GG}{\GD_0}{e_0}{\Gt_0}{C_0}$,
and for each $i$, 
$\con{C}_i : \forall \V{p}\V{a}. \: \V{\Gt_i} \pto{\V{\nu_i}} \con{D} \A \V{p} \A \V{a}$,
$\ITy{\GG,\V{x_i:\Gt_i[\Subst{p}{\pi_i},\Subst{a}{\Ga_i}]}}{\GD_i, \V{\Used{x_i}{M_i}}}{e_i}{\Gu_i}{C_i}$, 
$\Gt = \Gb$, 
$\GD = \pi_0 \GD_0 + \bigsqcup_i \GD_i$, and
\(C = C_0 \wedge \bigwedge_i ( C_i \wedge \Gb \eqty \Gu_i \wedge (\Gt_0 \eqty \con{D} \A \V{\pi_i} \A \V{\Ga_i}) \wedge \bigwedge_j M_{ij} \le \pi_0 \nu_{ij}[\Subst{p}{\pi_i}] )\).
Let $X_i$ be the set of unification variables introduced in the subderivation for each $e_i$, including $i = 0$. 
Then, we have $X = X_0 \uplus \biguplus_i X_i \uplus \set{\V{\Ga},\V{\Gp},\Gb}$.

By assumption, we have $\NTy{\GG\Gth'}{\GD'}{e}{\Gt'}{Q'}$.
This means we have $\NTy{\GG\Gth'}{\GD'_0}{e}{\Gt_0'}{Q'}$, and, for each $i$, 
\(\NTy{\GG\Gth',\V{x_i:\Gt_i[\V{\Subst{p}{\mu}},\V{\Subst{a}{\Gs}}]}}{\GD_i''}{e_i}{\Gu_i'}{Q'}\), 
$Q' \models \Gt_0 \eqty \con{D} \A \V{\mu} \A \V{\Gs}$, 
$Q' \models \GD_i'' = \GD'_i, \V{\Used{x_i}{\mu_0 \nu_i[\V{\Subst{p}{\mu}}]}}$, 
$Q' \models \GD' = \mu_0 \GD_0' + \bigsqcup_i \GD_i'$, and 
$Q' \models \Gt' \eqty \Gu_i'$.

By the induction hypothesis, for each $i$, there exists $\Gth_i$ such that 
$\dom(\Gth_i) \setminus \dom(\Gth') \subseteq X_i$, 
$(Q',\Gth_i)$ is a solution for $(\top,C_i)$,
$Q' \models \Gth_i |_{\dom(\Gth')} = \Gth'$, 
$Q' \models \Gu_i \Gth_i \eqty \Gu_i'$, and 
\(
Q' \models (\GD_i,\V{\Used{x_i}{M_i}}) \Gth_i \le (\GD_i', \V{\Used{x_i}{\mu_0 \nu_i[\V{\Subst{p}{\mu}}]}})
\).
Also, by the induction hypothesis, there exists $\Gth_0$ such that 
$\dom(\Gth_0) \setminus \dom(\Gth') \subseteq X_0$, 
$(Q',\Gth_0)$ is a solution for $(\top,C_0)$,
$Q' \models \Gth_0 |_{\dom(\Gth')} = \Gth'$, 
$Q' \models \Gt_0 \Gth_0 \eqty \Gt_0'$, and 
$Q' \models \GD_0 \Gth_0 \le \GD_0'$.
Let us define $\Gth$ as follows. 
\[
\Gth = \Gth' \uplus \Gth_0|_{X_0} \uplus \biguplus_i \Gth_i|_{X_i} \uplus \{ \VSubst{\pi}{\mu}, \VSubst{\Ga}{\Gs}, \Subst{\Gb}{\Gt'} \}
\]
Then, the rest of the discussion is similar to the other cases. \qed

 \fi

\end{document}